\documentclass{CSML}

\def\dOi{11(4:7)2015}
\lmcsheading%
{\dOi}
{1--32}
{}
{}
{Dec.~14, 2014}
{Dec.~10, 2015}
{}

\ACMCCS{[{\bf Theory of computation}]: Logic---modal and temporal logics}

\usepackage[numbers,sort]{natbib}
\usepackage[utf8]{inputenc}
\usepackage[theorems=false,paralist=false,colorlinks=true,linkstyle=frame]{generic}

\newcommand{\inv}[1]{\mspace{-1mu}\scalebox{0.5}[1.0]{$-\!$}#1}

\newcommand{\BUR}{\,\mspace{1.7mu}\raisebox{-1.6pt}{\includegraphics{figures/symbols.0}}\mspace{1.7mu}}
\newcommand{\BUL}{\,\mspace{1.7mu}\raisebox{-1.6pt}{\includegraphics{figures/symbols.1}}\mspace{1.7mu}}
\newcommand{\BLR}{\,\mspace{1.7mu}\raisebox{-1.6pt}{\includegraphics{figures/symbols.2}}\mspace{1.7mu}}
\newcommand{\BLL}{\,\mspace{1.7mu}\raisebox{-1.6pt}{\includegraphics{figures/symbols.3}}\mspace{1.7mu}}

\newcommand{\BULp}{\,\mspace{1.7mu}\raisebox{-1.6pt}{\includegraphics{figures/symbols.5}}\mspace{1.7mu}}
\newcommand{\BLRp}{\,\mspace{1.7mu}\raisebox{-1.6pt}{\includegraphics{figures/symbols.6}}\mspace{1.7mu}}

\newcommand{\BT}{\,\mspace{1.7mu}\raisebox{-1.6pt}{\includegraphics{figures/symbols.8}}\mspace{1.7mu}}
\newcommand{\BB}{\,\mspace{1.7mu}\raisebox{-1.6pt}{\includegraphics{figures/symbols.9}}\mspace{1.7mu}}
\newcommand{\BR}{\,\mspace{1.7mu}\raisebox{-1.6pt}{\includegraphics{figures/symbols.10}}\mspace{1.7mu}}
\newcommand{\BL}{\,\mspace{1.7mu}\raisebox{-1.6pt}{\includegraphics{figures/symbols.11}}\mspace{1.7mu}}
\newcommand{\B}{\,\mspace{1.7mu}\raisebox{-1.6pt}{\includegraphics{figures/symbols.12}}\mspace{1.7mu}}

\newcommand{\BTstrict}{\,\mspace{1.7mu}\raisebox{-1.6pt}{\includegraphics{figures/symbols.14}}\mspace{1.7mu}}
\newcommand{\BBstrict}{\,\mspace{1.7mu}\raisebox{-1.6pt}{\includegraphics{figures/symbols.15}}\mspace{1.7mu}}
\newcommand{\BRstrict}{\,\mspace{1.7mu}\raisebox{-1.6pt}{\includegraphics{figures/symbols.16}}\mspace{1.7mu}}

\newcommand{\DUR}{\,\raisebox{-2.7pt}{\includegraphics{figures/symbols.18}}}
\newcommand{\DUL}{\,\raisebox{-2.7pt}{\includegraphics{figures/symbols.19}}}
\newcommand{\DLR}{\,\raisebox{-2.7pt}{\includegraphics{figures/symbols.20}}}
\newcommand{\DLL}{\,\raisebox{-2.7pt}{\includegraphics{figures/symbols.21}}}
\newcommand{\DURp}{\,\raisebox{-2.7pt}{\includegraphics{figures/symbols.22}}}
\newcommand{\DULp}{\,\raisebox{-2.7pt}{\includegraphics{figures/symbols.23}}}
\newcommand{\DLRp}{\,\raisebox{-2.7pt}{\includegraphics{figures/symbols.24}}}
\newcommand{\DLLp}{\,\raisebox{-2.7pt}{\includegraphics{figures/symbols.25}}}

\newcommand{\DR}{\,\raisebox{-2.7pt}{\includegraphics{figures/symbols.28}}}

\newcommand{\D}{\,\raisebox{-2.7pt}{\includegraphics{figures/symbols.30}}}
\newcommand{\Darg}[1]{\,\mathbin{\ooalign{\raisebox{-2.7pt}{\includegraphics{figures/symbols.31}}\crcr\hidewidth\scriptsize$\mspace{-5mu}#1$\hidewidth}}}
\newcommand{\Dargsmall}[1]{\,\mathbin{\ooalign{\raisebox{-2.5pt}{\includegraphics[scale=0.875]{figures/symbols.31}}\crcr\hidewidth\tiny$\mspace{-5mu}#1$\hidewidth}}}
\newcommand{\DTstrict}{\,\raisebox{-2.7pt}{\includegraphics{figures/symbols.32}}}
\newcommand{\DBstrict}{\,\raisebox{-2.7pt}{\includegraphics{figures/symbols.33}}}
\newcommand{\DRstrict}{\,\raisebox{-2.7pt}{\includegraphics{figures/symbols.34}}}
\newcommand{\DLstrict}{\,\raisebox{-2.7pt}{\includegraphics{figures/symbols.35}}}

\newcommand{\CUR}{\,\raisebox{-2.2pt}{\includegraphics{figures/symbols.36}}\,}
\newcommand{\CUL}{\,\raisebox{-2.2pt}{\includegraphics{figures/symbols.37}}\,}
\newcommand{\CLR}{\,\raisebox{-2.2pt}{\includegraphics{figures/symbols.38}}\,}
\newcommand{\CLL}{\,\raisebox{-2.2pt}{\includegraphics{figures/symbols.39}}\,}
\newcommand{\CURp}{\,\raisebox{-2.2pt}{\includegraphics{figures/symbols.40}}\,}
\newcommand{\CULp}{\,\raisebox{-2.2pt}{\includegraphics{figures/symbols.41}}\,}
\newcommand{\CLRp}{\,\raisebox{-2.2pt}{\includegraphics{figures/symbols.42}}\,}
\newcommand{\CLLp}{\,\raisebox{-2.2pt}{\includegraphics{figures/symbols.43}}\,}

\newcommand{\Carg}[1]{\,\mathbin{\ooalign{\raisebox{-2.2pt}{\includegraphics{figures/symbols.49}}\crcr\hidewidth\scriptsize$\mspace{-3mu}#1$\hidewidth}}\,}
\newcommand{\Cargsmall}[1]{\,\mathbin{\ooalign{\raisebox{-2pt}{\includegraphics[scale=0.875]{figures/symbols.49}}\crcr\hidewidth\tiny$\mspace{-3mu}#1$\hidewidth}}\,}
\newcommand{\CTstrict}{\,\raisebox{-2.2pt}{\includegraphics{figures/symbols.50}}\,}
\newcommand{\CBstrict}{\,\raisebox{-2.2pt}{\includegraphics{figures/symbols.51}}\,}
\newcommand{\CRstrict}{\,\raisebox{-2.2pt}{\includegraphics{figures/symbols.52}}\,}
\newcommand{\CLstrict}{\,\raisebox{-2.2pt}{\includegraphics{figures/symbols.53}}\,}

\newcommand{\fa}[1]{\forall{\;#1}.\;}
\newcommand{\ex}[1]{\exists{\;#1}.\;}

\newcommand{\AG}{\,\mathbf{AG}\,}

\newcommand{\AF}{\,\mathbf{AF}\,}
\newcommand{\AX}{\,\mathbf{AX}\,}
\newcommand{\EX}{\,\mathbf{EX}\,}

\newcommand{\settc}[2]{{\{ #1 \,:\, #2 \}}}
\newcommand{\bigsettc}[2]{{\bigl\{ #1 \,:\,#2 \bigr\}}}

\newcommand{\ang}[1]{{\langle #1 \rangle}}
\newcommand{\bigang}[1]{{\bigl\langle #1 \bigr\rangle}}

\newcommand{\len}[1]{{\lvert #1 \rvert}}

\newcommand{\prj}[2]{\mspace{2mu}\downarrow_{#1}\mspace{-5mu}{\ifx\\#2\\\else\!\!\;(#2)\fi}}

\newcommand{\closure}{\mathsf{closure}}

\newcommand{\type}{\mathsf{type}}
\newcommand{\req}{\mathsf{-req}}
\renewcommand{\obs}{\mathsf{-obs}}

\newcommand{\tree}{\mathsf{tree}}
\renewcommand{\path}{\mathsf{path}}
\newcommand{\atom}{\mathsf{atom}}
\newcommand{\cluster}{\mathsf{cluster}}
\newcommand{\positive}{\mathsf{pos}}
\newcommand{\negative}{\mathsf{neg}}
\newcommand{\singleton}{\mathsf{sing}}
\newcommand{\equalin}{\stackrel{\in}{=}}

\newcommand{\refmatch}[1]{\hyperref[#1]{M\ref*{#1}}}
\newcommand{\reffulfill}[1]{\hyperref[#1]{F\ref*{#1}}}
\newcommand{\refglobal}[1]{\hyperref[#1]{G\ref*{#1}}}

\makeatletter
\renewcommand\bibsection{
  \section*{\refname\@mkboth{\MakeUppercase{\refname}}{\MakeUppercase{\refname}}}
}
\makeatother

\begin{document}

\title[A decidable weakening of Compass Logic]      {\texorpdfstring{A decidable weakening of Compass Logic \\based on cone-shaped cardinal directions}      {A decidable weakening of Compass Logic based on cone-shaped cardinal directions\rsuper*}}

\author[A.~Montanari]{Angelo Montanari\rsuper a}
\address{{\lsuper a}University of Udine, Udine, Italy}
\email{angelo.montanari@uniud.it}

\author[Puppis]{Gabriele Puppis\rsuper b}
\address{{\lsuper b}LaBRI / CNRS, Bordeaux, France}
\email{gabriele.puppis@labri.fr}

\author[Sala]{Pietro Sala\rsuper c}
\address{{\lsuper c}University of Verona, Verona, Italy}
\email{pietro.sala@univr.it}

\titlecomment{{\lsuper *}A short preliminary version of this work appeared in~\cite{cone_logic}.}

\begin{abstract}
We introduce a modal logic, called Cone Logic, whose formulas describe properties of points in the plane and spatial relationships between them. Points are labelled by proposition letters and spatial relations are induced by the four cone-shaped cardinal directions. Cone Logic can be seen as a weakening of Venema's Compass Logic. We prove that, unlike Compass Logic and other projection-based spatial logics, its satisfiability problem is decidable (precisely, \pspace-complete). 
We also show that it is expressive enough to capture meaningful interval temporal logics -- in particular, the interval temporal logic of Allen's relations `Begins', `During', and `Later', and their transposes. 
 \end{abstract}

\maketitle

\section{Introduction}\label{sec:introduction}

Spatial reasoning has both a strong theoretical relevance and many applications in various areas of computer 
science, including robotics, natural language processing, and geographical information systems 
\cite{spatial_logic_handbook,spatial_and_temporal_reasoning,many_dimensional_modal_logics}. 
However, despite the widespread interest in the topic, few 
techniques have been developed to automatically (and efficiently) reason about 
spatial relations over infinite structures. As a matter of fact, spatial reasoning 
has been mainly investigated in quite restricted algebraic settings.

Most logical formalisms for spatial reasoning can be conveniently classified into two classes, 
on the basis of the type of relations they make use of.
On the one side, there are logics whose modalities are based on cardinal directions.
The most notable example of a formalism in this class is Venema's Compass Logic 
\cite{compass_logic}, which allows one to express properties such as: 
``\emph{from every point labelled with $a$ there is a point to the north of it, that is, above 
it and vertically aligned to it, that is labelled with $b$}''. 
On the other side, there are formalisms based on topological relations, 
like the Region Connection Calculus \cite{region_spatial_logic,point-set-topological-spational-relations}, 
which can express properties such as:
``\emph{two regions of points, labelled with $a$ and $b$ respectively, are externally connected, that is, tangent}''.
A quite extensive discussion of the expressiveness of various spatial logics and of their 
connections can be found in \cite{modal_logics_topology}.

In this paper, we introduce a novel spatial modal logic, called \emph{Cone Logic}, 
which allows one to reason about directional relations between points in the rational plane. 
Being based on cardinal directions, our logic falls inside the first group of formalisms
discussed above. However, unlike most logics based on cardinal directions, 
the modal operators of Cone Logic range over cone-shaped regions of the plane -- formally, 
over quadrants -- rather than semi-axes.
To stress this difference, we will often talk of \emph{cone-shaped cardinal directions}, 
as opposed to \emph{projection-based cardinal directions} (see Figure \ref{fig:spatial-relations}).
This difference is also reflected in considerably better algorithmic properties.
While the satisfiability problem for modal logics with projection-based cardinal directions
-- notably, Compass Logic -- 
turns out to be highly undecidable \cite{undecidability_compass_logic,spatial_pnl},
we prove that Cone Logic enjoys a decidable satisfiability problem (in fact, \pspace-complete)
by making use of a suitable filtration technique. 
We also show that Cone Logic subsumes interesting interval temporal logics such as 
the temporal logic of sub-intervals/super-intervals, thus generalizing previous results 
in the literature \cite{subinterval_tableau_journal} and basically disproving 
a conjecture by Lodaya \cite{undecidability_BE}.

\medskip\noindent
{\bf Related work.}
The paper that is most related in spirit to the present work is that of Venema \cite{compass_logic},
who studies \emph{Compass Logic}.
Compass Logic is a two-dimensional modal logic interpreted over the Cartesian product of 
two linear orders, which features two pairs of modalities, each pair ranging over one of the two orders. 
The first undecidability result for the satisfiability problem of Compass Logic was shown in \cite{undecidability_compass_logic}
and it covers both the case where the logic is interpreted over the discrete infinite grid $\bbN\times\bbN$ 
and the case where the logic is interpreted over the Euclidean space $\bbR\times\bbR$.
In \cite{products_of_linear_modal_logics}, similar formalisms based on products of two linear modal logics
have been studied and the above-mentioned undecidability results have been strengthened to cover practically all 
classes of products of infinite/unbounded linear orders.
These negative results stem from the possibility of encoding halting computations of Turing machines
inside a two-dimensional structure and expressing the correctness of the encoding in the logic.

Cone Logic can be viewed as the fragment of Venema's Compass Logic
obtained from the full logic by enforcing the following restriction:
quantifications along one axis can be used only after a similar quantification along the 
other axis. 
Such a constrain makes it impossible to correctly encode computations of Turing machines in the 
underlying two-dimensional space, thus leaving room to recover the decidability of the satisfiability 
problem.

There is also a tight connection between modal logics over two-dimensional spaces 
and fragments of Halpern and Shoham's modal logic of time intervals (HS) \cite{interval_modal_logic}.
According to such a correspondence, intervals over 
a linearly ordered temporal domain are interpreted as points over a two-dimensional space. 
In Section \ref{sec:applications}, we will show how such a correspondence can be lifted 
to the logical level, by reducing the satisfiability problem for an expressive fragment 
of HS to the satisfiability problem for (a subset of formulas of) Cone Logic.

Other multi-dimensional spatial logics are studied in 
\cite{two_dimensional_logic,qualitative_spatial_modal_logics,multi_dimentional_logics}
(with different goals in mind). 
Some of them retain good decidability properties, but their expressive power 
is often limited. An example is
the logic proposed by Bennett in \cite{qualitative_spatial_modal_logics},
which uses a single modal operator interpreted as the interior in a given topology.
This logic is essentially equivalent to S4 and its satisfiability problem is \pspace-complete.

\medskip\noindent
{\bf Structure of the paper.}
In Section \ref{sec:logic}, we define 
syntax and 
semantics of Cone Logic and we discuss its expressiveness and satisfiability problem. 
In Section  \ref{sec:equivalences}, we introduce the basic machinery for attacking the 
satisfiability problem. In Section \ref{sec:dectree}, we show how to turn
a labelled region of the rational plane into an infinite (decomposition) tree 
structure. Then, in Section \ref{sec:pseudomodel}, we prove a tree
(pseudo-)model property 
for Cone Logic, that is, we describe models of
satisfiable Cone Logic formulas by means of suitable labelled tree structures. 
In Section \ref{sec:solution}, we exploit such a tree model property to
reduce the satisfiability problem for Cone Logic to the satisfiability
problem for a simple fragment of CTL. In Section \ref{sec:applications},
we make use of such a decidability result to prove that a meaningful
fragment of Halpern and Shoham's interval temporal logic HS,
interpreted over dense linear structures, is decidable in polynomial space.
In Section \ref{sec:conclusions}, we make some final remarks and we 
discuss related and open problems.

 \section{The logic}\label{sec:logic}

In this paper, we generically denote by $\bbP$ either the \emph{rational plane} $\bbQ\times\bbQ$ or the \emph{real}
(Euclidean) \emph{plane} $\bbR\times\bbR$. We will define the semantics of formulas of Cone Logic in the same way over labellings of the rational plane and labellings of the real plane. 

We call \emph{spatial relation} any binary relation $\Carg{d}\in\bbP\times\bbP$ between points in the plane. We 
use the infix notation $p \Carg{d} q$ for saying that two points $p,q\in\bbP$ satisfy a given spatial relation $\Carg{d}$. 
We start by defining some \emph{basic spatial relations}, denoted $\CTstrict$, $\CBstrict$, $\CRstrict$, $\CLstrict$,
that correspond to the four projection-based cardinal directions `North', `South', `East' and `West' (see  
Figure \ref{fig:spatial-relations} - left):
$$
\begin{array}{rclrcl}
  (x,y) ~\CTstrict~ (x',y') &~\text{iff}& ~ x=x' \et y<y'   & \quad\quad
  (x,y) ~\CBstrict~ (x',y') &~\text{iff}& ~ x=x' \et y>y'   \\[2ex]
  (x,y) ~\CRstrict~ (x',y') &~\text{iff}& ~ x<x' \et y=y'   & \quad\quad
  (x,y) ~\CLstrict~ (x',y') &~\text{iff}& ~ x>x' \et y=y'.
\end{array}
$$
Using the above basic relations and set-theoretic operations, one can construct new spatial relations. 
We define the composition of two spatial relations $\Carg{d}$ and $\Carg{e}$  by
$\Carg{d}\Carg{e}=\bigsettc{(p,r)}{\ex{q\in\bbP}~(p,q)\in\Carg{d} \et (q,r)\in\Carg{e}}$.
We are interested in the following spatial relations:
$$
\begin{array}{rclrcl}
  \CULp &~=&~ \CTstrict\CLstrict\cup\CLstrict   & \qquad\qquad\qquad 
  \CURp &~=&~ \CTstrict\CRstrict\cup\CRstrict \\[2ex]
  \CLLp &~=&~ \CBstrict\CLstrict\cup\CLstrict   & \qquad\qquad\qquad 
  \CLRp &~=&~ \CBstrict\CRstrict\cup\CRstrict
\end{array}
$$

Observe that, up to a rotation of the axes, the derived relations $\CUL=\CTstrict\CULp$, $\CUR=\CTstrict\CURp$, 
$\CLL=\CBstrict\CLLp$, and $\CLR=\CBstrict\CLRp$ can be viewed as the four cone-shaped cardinal relations 
`North', `East', `West' and `South' \cite{cardinal_directions} (see Figure \ref{fig:spatial-relations} - right). 
\begin{figure}[!!t]
\centering
\includegraphics[scale=0.8]{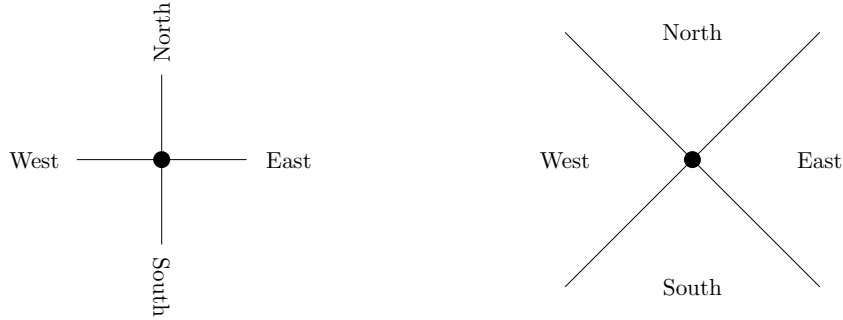}
\caption{Projection-based and cone-shaped cardinal directions.}
\label{fig:spatial-relations}
\end{figure}

\medskip
We introduce Cone Logic as a modal logic based on specific spatial relations. As such, it can express 
properties of single elements 
(the labels associated with the points in a plane) and binary relationships between elements
(it admits existential quantifications over points that satisfy a given spatial relation). 
The modal operators of Cone Logic are induced by the {\sl six spatial relations}
$\CTstrict$, $\CBstrict$, $\CULp$, $\CURp$, $\CLLp$, $\CLRp$ (the reason for such a choice will become evident
in the following). Unless otherwise specified, hereafter the term ``spatial relation'' will always refer 
to one of these six relations.

Given a set $\Sigma$ of proposition letters, formulas of Cone Logic are built up from $\Sigma$ using the Boolean connectives 
$\neg$ and $\vel$ and the existential modalities that correspond to the six spatial relations:
\begin{align*}
  \varphi &\:\:\::=\:\:\: a \tag{$\forall a\in\Sigma$} \\
          &\:\:\:\:\:\:\:||\:\:\: \neg\varphi' \:\:\:||\:\:\: \varphi'\vel\varphi'' \\
          &\:\:\:\:\:\:\:||\:\:\: \!\!\Darg{d}\varphi' \qquad\qquad 
          \tag{$\forall \Carg{d}\in\big\{\CTstrict, \CBstrict, \CULp, \CURp, \CLLp, \CLRp\big\}$}
\end{align*}
We evaluate Cone Logic formulas over labellings of the plane or (sub)regions of it, starting from an initial point.
Precisely, our models are structures of the form $\ang{P,(R_a)_{a\in\Sigma},p}$, where $P\subseteq\bbP$, $R_a\subseteq P$
for all $a\in\Sigma$, and $p\in P$. The formal semantics is defined as follows:
\begin{itemize}
  \item for all proposition letters $a\in\Sigma$, $\ang{P,(R_a)_{a\in\Sigma},p} \sat a$ iff $p\in R_a$,
  \item $\ang{P,(R_a)_{a\in\Sigma},p} \sat \neg\varphi'$ iff $\ang{P,(R_a)_{\Sigma},p} \not\sat\varphi'$,  
  \item $\ang{P,(R_a)_{a\in\Sigma},p} \sat \varphi' \vel \varphi''$ iff $\ang{P,(R_a)_{a\in\Sigma},p} \sat \varphi'$ or $\ang{P,(R_a)_{a\in\Sigma},p} \sat \varphi''$,
  \item for all spatial relations $\Carg{d}$, $\ang{P,(R_a)_{a\in\Sigma},p} \sat \Darg{d}\varphi'$ iff $\ang{P,(R_a)_{a\in\Sigma},q} \sat\varphi'$ for some
        point $q\in P$ such that $p \Carg{d} q$.
\end{itemize}
We will freely use shorthands like $\varphi'\et\varphi''=\neg(\neg\varphi' \vel \neg\varphi'')$, 
$\bot=a\et\neg a$, $\DUL\varphi=\DTstrict\DULp\varphi$, $\BUL\varphi=\neg\DUL\neg\varphi$, $\BT\varphi=\BUL\BUR\varphi$, $\B\varphi=\BUL\BLR\varphi$, and so on.

\medskip
Cone Logic is well-suited for expressing spatial relationships between points, curves, and regions 
over the plane. Below, we give an intuitive account of its expressiveness through a couple of examples. 

\begin{exa}\label{ex:rectangle}
To begin with, we show how to define an $a$-labelled open rectangular region, whose edges are aligned 
with the $x$- and $y$-axes (see Figure \ref{fig:rectangle}):
$$
\begin{array}{rcl}
  \varphi &~=&~   \D a \et \D b \et \D c \et \D d \et \D e \\[1ex]
          &~\et&~ \B (a \then \DUR a \et \DLL a) \et 
                  \B (\neg a \iff b \!\vel\! c \!\vel\! d \!\vel\! e)  \\[1ex]
          &~\et&~ \B (b \then \BL b) \et \B (c \then \BR c) \et \B (d \then \BB d) \et \B (e \then \BT e).
\end{array}
$$
\end{exa}

\begin{figure}[!!t]
\centering
\includegraphics[scale=1]{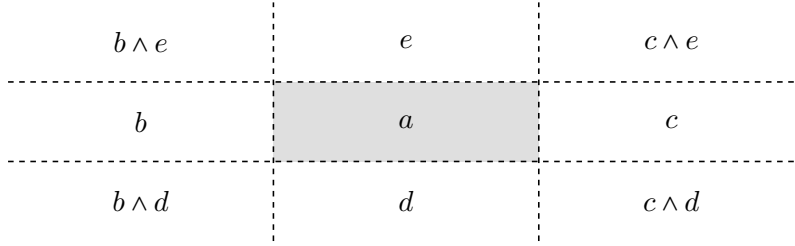}
\caption{An $a$-labelled open rectangle.}
\label{fig:rectangle}
\end{figure}

\begin{exa}\label{ex:cantor}
The second example uses the derived operators $\DUR$ and $\DLL$ to enforce non-trivial spatial relationships 
between labelled regions of the {\sl rational} plane. Let $\Sigma$ be an alphabet containing $n+2$ proposition
letters $a,b_1,...,b_n,c$ and let $<$ be the partial order over $\Sigma$ such that $a<b_i<c$, for all $1\le i\le n$, and $b_i\not<b_j$, for all $1\le i ,j\le n$ with $i\neq j$. 
As usual, we write $a\le b$ (resp., $a\ge b$) if $a=b$ or $a<b$ (resp., $a>b$). Consider now the formula $\varphi_\le$ defined as follows:
$$
\begin{array}{rcl}
  \varphi_{\le} &=&     \B\!\displaystyle\bigvee\limits_{d\in\Sigma}\!d ~~\et~~ 
                        \B\!\displaystyle\bigwedge\limits_{d\neq e}\!\neg(d\et e)                    ~\et~ \B\!\displaystyle\bigwedge\limits_{d\in\Sigma}
                        \Bigl(
                          d ~\then~ \displaystyle\bigwedge\limits_{e\ge d}\!\!\DUR e 
                                    ~\et
                                    \BUR\!\bigvee\limits_{e\ge d}\!e
                                    ~\et
                                    \displaystyle\bigwedge\limits_{e\le d}\!\!\DLL e
                                    ~\et
                                    \BLL\!\bigvee\limits_{e\le d}\!e
                        \Bigr) \ .
\end{array}
$$
The unique (up to homomorphism) labelling of the {\sl rational} plane $\bbQ\times\bbQ$ that satisfies $\varphi_\le$ 
is depicted in Figure \ref{fig:cantor}. Notice that each $b_i$-labelled region is an infinite union of disjoint 
open rectangles (the coordinates of their corners are given by pairs of irrational numbers, which, of 
course, do not belong to the rational plane). Moreover, the $b_i$-labelled open rectangles are arranged 
densely in the rational plane, that is, for all $1\le i,j,k\le n$, with $i\neq j$, all $b_i$-labelled points $(x_i,y_i)$, 
and all $b_j$-labelled points $(x_j,y_j$), with $x_i<x_j$ and $y_i>y_j$, there is a $b_k$-labelled point 
$(x_k,y_k)$ such that $x_1<x_k<x_2$ and $y_1>y_k>y_2$. 
We also observe that the formula $\varphi_\le$ cannot be satisfied by any labelling of the {\sl real} plane 
$\bbR\times\bbR$. Indeed, $\varphi_\le$ requires that the subregions $R_a,R_{b_1},...,R_{b_n},R_c$ are ``open'' 
(in the sense that they do not contain points on their boundaries) and they form a partition of the plane:
this is against the assumption that the plane is compact, as in this case boundaries would be covered by 
the subregions.
\end{exa}

\begin{figure}[!!t]
\centering
\includegraphics[scale=0.8]{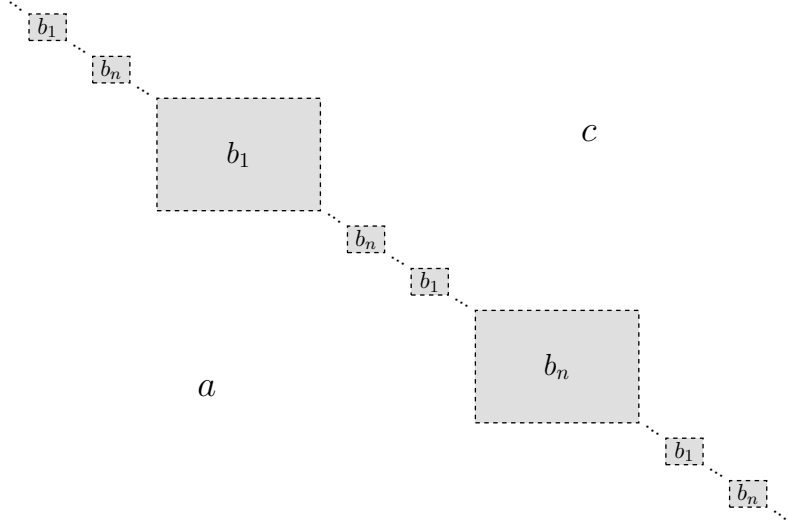}
\caption{A labelled rational plane satisfying $\varphi_\le$.}
\label{fig:cantor}
\end{figure}

In the following, we focus our attention on the \emph{satisfiability problem} for Cone Logic, which
consists of deciding whether a given formula $\varphi$ holds at some point of a labelled region of the 
(rational or real) plane. In particular, we are interested in satisfiability of formulas interpreted 
over rectangular regions of the form $X\times Y$, where $X$ and $Y$ are open or closed 
intervals\footnote{Here the term ``interval'' is used as a synonym for {\sl convex subset}. We accordingly 
                   denote intervals by $[x,y]$, $(x,y)$, $[x,y)$, $(x,y]$, where a bracket is square 
                   or round depending on whether the corresponding endpoint is included or not in the interval.}
of $\bbQ$ (resp., $\bbR$).
Before describing our decision procedure for the satisfiability problem for Cone Logic, we make a few remarks.

\begin{rem}\label{rem:rationalplanes}
In Example \ref{ex:cantor}, we showed that there exist Cone Logic formulas that can only be satisfied 
over dense non-Euclidean (e.g., rational) planes. Here, we prove that the converse does not hold, namely, 
that every formula of Cone Logic that is satisfied in some (rational or real) plane is also satisfiable 
in the rational plane. First of all, we observe that Cone Logic can be viewed as a fragment of classical 
first-order logic that uses pairs of elements of the underlying domain to denote points, some binary 
relations to represent their labels, and a (definable) dense linear order to describe the spatial relations. 
As an example, a Cone Logic formula of the form:
$$
  \varphi ~=~ \DURp a
$$
can be translated into the following, equi-satisfiable first-order formula:
$$
\begin{array}{rl}
  \tilde{\varphi}(x,y)
  &=~ \text{\it ``\,$<$ is a dense linear order with neither a minimal nor a maximal element''} \\
  &   \et\:\:\: \ex{x',y'} \:\: x<x' ~\et~ (y=y' \vel y<y') ~\et~ R_a(x',y').
\end{array}
$$
According to the above translation, if $\tilde{\varphi}(x,y)$ holds in some structure 
$\ang{L,(R_a)_{a\in\Sigma},<,i,j}$, where $R_a$ and $<$ are binary relations on the domain $L$
and $i,j$ are elements of $L$, then $(L,<)$ is a dense linear order with neither a minimal nor a maximal 
element and $\varphi$ holds in the labelled plane $\ang{L\times L,(R_a)_{a\in\Sigma},(i,j)}$. As $\tilde{\varphi}(x,y)$ 
is a first-order formula, it follows from L\"owenheim-Skolem theorem that,  without loss of generality, $L$ can be 
assumed to be countable. Finally, since $(\bbQ,<)$ is up to isomorphism the only countable dense linear order with 
neither a minimal nor a maximal element, we conclude that $\varphi$ is satisfied by a labelling of the rational plane. 
\end{rem}

\begin{rem}\label{rem:stripes}
Recall that the rational (resp., real) plane is homomorphic to any open rectangular subregion of it of
the form $X\times Y$, with $X=(x_0,x_1)$ and $Y=(y_0,y_1)$ open intervals. This means that, for the purpose 
of studying satisfiability of Cone Logic, it does not matter if we consider labellings of the entire plane 
or labellings of open rectangular subregions of it. Similarly, the complexity of the satisfiability 
problem does not change if we consider {\sl closed} rectangles. Indeed, any formula $\varphi$ of 
Cone Logic, interpreted over a region of the form $X\times Y$, where $X=(x_0,x_1)$ is an open interval, 
can be rewritten into an equi-satisfiable formula $\bar{\varphi}$, interpreted over the region 
$\bar{X}\times Y$, where $\bar{X}=[x_0,x_1]$ is a closed, non-singleton interval, and vice versa. 
As an example, the Cone Logic formula 
$$
  \varphi ~=~ \DURp a
$$ 
interpreted over a labelling of $(x_0,x_1)\times\bbQ$ can be rewritten as
$$
  \bar{\varphi} ~=~ \B\big(\BR\bot \vel \BL\bot \then a_\bot\big) \et \DURp a
$$
which is interpreted over a labelling of $[x_0,x_1]\times\bbQ$ 
(the idea is that points along left and right boundaries are labelled 
with the fresh proposition letter $a_\bot$).
\end{rem}

Thanks to the above two remarks, we can restrict our attention to satisfiability of Cone Logic
formulas over specific regions of the plane, called stripes.

\begin{defi}\label{def:stripes}
A \emph{stripe} is a 
region
of the form $X\times\bbQ$, 
where $X=[x_0,x_1]$ is a closed non-singleton interval.  
\end{defi}

The relationships between the Cone Logic and other 
two-dimensional modal logics
deserve a little discussion.\ 
Many logics interpreted over two-dimensional structures make use of projection-based modalities,
that is, modalities induced by the accessibility relations along the two orthogonal axes.
Compass Logic \cite{compass_logic} is the most notable example of these two-dimensional logics, as it
comprises the four modalities $\DTstrict$, $\DBstrict$, $\DRstrict$, $\DLstrict$, 
allowing one to move along one of the two coordinates while keeping the other coordinate constant. 
As we have already seen, modalities based on cone-shaped 
cardinal
directions can be easily 
defined in terms of projection-based modalities, e.g., $\DUR \varphi = \DTstrict\DRstrict\varphi$. 
Cone Logic can thus be viewed as a fragment of Compass Logic. However, Cone Logic inherits 
from Compass Logic only some desirable features. 
For instance, suppose that one is interested in constraining a given proposition letter $a$ to occur 
along the positive $x$-axis, and possibly somewhere else. Such a condition can be easily forced in
Compass Logic by means of the formula $\DRstrict a$. 
Cone Logic can enforce a similar constraint by means of the formula
$$
  \B(y_0 \vel \DTstrict y_0 \vel \DBstrict y_0) \et \BT \neg y_0 \et \BB \neg y_0 \et \DR (a \et y_0), \phantom{\ .~}
$$
where $y_0$ is a fresh proposition letter.
Similarly, the Compass Logic formula $\BRstrict a$ can be expressed in Cone Logic as follows:
$$
  \B(y_0 \vel \DTstrict y_0 \vel \DBstrict y_0) \et \BT \neg y_0 \et \BB \neg y_0 \et \BR (y_0 \then a).
$$
It is worth noticing, however, that the above translations can be performed at the cost 
of introducing additional labels
-- e.g.,~$y_0$ -- that can only appear along specific axes.
Hence, only boundedly many constraints of the above forms can be enforced within a single formula of Cone Logic.
We will see that such a limitation (weakening) can be traded for a positive decidability result.

 \section{Basic machinery: types, dependencies, clusters, and shadings}\label{sec:equivalences}

From now on, we refer to a generic formula $\varphi$ of Cone Logic. The basic idea underlying the decision
procedure for the satisfiability of $\varphi$ is to first look at how the spatial constraints defined 
by the subformulas of $\varphi$ can be satisfied locally over the points of the plane and then to propagate 
these constraints to larger and larger regions of the plane.
Below, we introduce some key concepts that ease such an analysis.

\begin{defi}\label{def:atom}
Let $\varphi$ be a formula of Cone Logic. The closure of $\varphi$, denoted by $\closure(\varphi)$,
is  the set of all subformulas of $\varphi$ and of their negations (we identify any subformula $\neg\neg\alpha$ 
with $\alpha$). A \emph{$\varphi$-atom} is a non-empty set 
$A\subseteq\closure(\varphi)$ such that:
\begin{itemize}
  \item for every formula $\alpha\in\closure(\varphi)$, $\alpha\in A$ iff $\neg\alpha\nin A$, 
  \item for every formula $\gamma=\alpha\vel\beta\in\closure(\varphi)$, $\gamma\in A$ iff $\alpha\in A$ or $\beta\in A$.
\end{itemize}
\end{defi}

\noindent
Note that the cardinality of $\closure(\varphi)$ is {\sl linear} in the size 
$\len{\varphi}$ of $\varphi$, while the number of $\varphi$-atoms is at most {\sl exponential} in $\len{\varphi}$. 

Let $\cP=\ang{P,(R_a)_{a\in\Sigma}}$ be a labelled region. We associate with each point $p$ in $P$ 
the set of all formulas $\alpha\in\closure(\varphi)$ such that $\ang{\cP,p}\sat\alpha$. Such a set is called 
the \emph{$\varphi$-type} of $p$ and it is denoted by $\type_\cP(p)$. 
It can be easily checked  that each $\varphi$-type is a $\varphi$-atom, but not vice versa.

\smallskip

Given a $\varphi$-atom $A$ and a spatial relation $\Carg{d}$, we denote by $\Darg{d}\req(A)$ the set of all 
formulas $\alpha\in\closure(\varphi)$ such that $\Darg{d}\alpha\in A$.These formulas can be thought of as  
the \emph{requests} of $A$ along the direction $\Carg{d}$. Similarly, we denote by $\Darg{d}\obs(A)$ the 
set of all formulas $\alpha\in A$ such that $\Darg{d}\alpha\in\closure(\varphi)$. These formulas can be thought
of as the \emph{observables} of $A$ along the direction $\Carg{d}$. Making use of these sets, we can 
associate with each spatial relation $\Carg{d}$ a corresponding relation between $\varphi$-atoms
(with a little abuse of notation, we denote it by $\Carg{d}$).

\begin{defi}\label{def:dependencies} 
Let $\Carg{d}\in\big\{\CTstrict,\CBstrict,\CULp,\CURp,\CLLp,\CLRp\big\}$ 
be a spatial relation and $A,B$ be two $\varphi$-atoms. 
We write $A \Carg{d} B$ if and only if it holds that:
$$
\begin{array}{lll}
  \Darg{e}\req(A) &~\supseteq& ~\Darg{e}\obs(B) \,\cup \Darg{e}\req(B); \\[1ex]
  \Darg{e\!'}\req(B) &~\supseteq& ~\Darg{e\!'}\obs(A) \,\cup \Darg{e\!'}\req(A),
\end{array}
$$
for all spatial relations $\Carg{e}\in\big\{\CTstrict,\CBstrict,\CULp,\CURp,\CLLp,\CLRp\big\}$ such that 
$\Carg{e}\supseteq\Carg{e}\Carg{d}$ (in particular, for $\Carg{e}=\Carg{d})$, where $\Carg{e\!'}$ is the 
inverse of $\Carg{e}$.
\end{defi}

It is worth looking at some concrete examples of the above definition. For instance, let $\Carg{d}=\CURp$ and
observe that $\Carg{e}\supseteq\Carg{e}\Carg{d}$ only if $\Carg{e}=\CURp$. In this case, the definition
amounts at saying that $A \CURp B$ iff all requests and observables of $B$ along the direction $\CURp$ 
are also requests of $A$ along $\CURp$, and, symmetrically, all requests and observables of $A$ along 
$\CLLp$ are also requests of $B$ along $\CLLp$. Let us now consider the more interesting case of $\Carg{d}=\CTstrict$. 
Here we have $\Carg{e}\supseteq\Carg{e}\Carg{d}$ iff $\Carg{e}\in\big\{\CTstrict,\CULp,\CURp\big\}$. 
In particular, we can write $A \CTstrict B$ only if the requests and the observables of $B$ along 
the direction $\CTstrict$ (resp., $\CULp$, $\CURp$) are also requests of $A$ along the direction $\CTstrict$ 
(resp., $\CULp$, $\CURp$), and symmetrically for the inverses $\CBstrict$, $\CLLp$, and $\CLRp$.

\smallskip

We conclude this short section with a few important remarks. First, we observe that the above-defined 
relations on $\varphi$-atoms are transitive, e.g., $A_1 \CURp A_2 \CURp A_3$ implies $A_1 \CURp A_3$, 
and have inverses (e.g., $A \CURp B$ iff $B \CLLp A$), exactly as the corresponding relations on points. 
Moreover, they satisfy some natural compositional properties, e.g., $A \CURp B \CTstrict C$ implies $A 
\CURp C$. The most important property, however, is the following one, which is called \emph{view-to-type 
dependency}: for all points $p,q$ of $\cP$ and all spatial relations $\Carg{d}$, 
$$
  p \Carg{d} q   \qquad\text{implies}\qquad  \type_\cP(p) \Carg{d} \type_\cP(q)
$$
(note that the converse implication does not hold).

\smallskip

The above notions can be easily extended to {\sl sets} of atoms (these sets are meant to represent sets of 
types of points in a region of the plane). First, we define a \emph{$\varphi$-cluster} as any non-empty set $C$ 
of $\varphi$-atoms. Then, for a cluster $C$ and a spatial relation $\Carg{d}$, we denote by $\Darg{d}\req(C)$ 
and $\Darg{d}\obs(C)$, respectively, the set $\bigcup_{A\in C}\Darg{d}\req(A)$ and the set $\bigcup_{A\in C}\Darg{d}\obs(A)$. Moreover, given two $\varphi$-clusters $C,D$, we write $C \Carg{d} D$ whenever $A 
\Carg{d} B$ holds for all  $A\in C$ and all $B\in D$. Finally, we associate with each non-empty subregion 
$R$ of $\cP$ its \emph{$\varphi$-shading}, which is defined as the set $\type_\cP(R)=\bigsettc{\type_\cP(p)}
{p\in R}$ and consists of all $\varphi$-types of points of $R$. 
Clearly, the formula $\varphi$ is satisfied at some point $p$ of $P$ if and only if f the shading $\type_\cP(P)$ 
contains an atom $A$ such that $\varphi\in A$. Hereafter, we shall omit the argument $\varphi$ from the 
terminology and notation so far introduced, thus calling a $\varphi$-atom (resp., $\varphi$-type, 
$\varphi$-cluster, etc.) simply an atom (resp., type, cluster, etc.).
 \section{From the plane to the binary tree}\label{sec:dectree}

In this section, we introduce a suitable notion of decomposition of a labelled region of the rational 
plane (more precisely, a labelled stripe) and we iteratively apply it in order to obtain an infinite 
decomposition tree structure that faithfully represents the original model. Then, in the next section, 
we make use of such a decomposition to establish a tree (pseudo-)model property for the satisfiable 
formulas of Cone Logic. 

\subsection{Profiles and stripe expressions}\label{subsec:profiles}

To start with, we consider the types along vertical lines of a labelled plane:

\begin{defi}\label{def:profile}
A \emph{profile} is a non-empty finite sequence $S$ of atoms and clusters 
such that, for every $1\le i\le\len{S}$, if $S(i)$ is an atom, then $1<i<\len{S}$ and both 
$S(i-1)$ and $S(i+1)$ are clusters. 
\end{defi}

We will use profiles to represent the arrangement of the types along a certain vertical line of
the labelled plane. The general idea is that one can partition the vertical line into a finite
sequence of contiguous open 
or
singleton segments in such a way that the shading of each open segment
(resp., the type of each singleton segment) coincides with the cluster (resp., atom) 
at some specific position
of the profile. As an example, Figure \ref{fig:profiles}(a) depicts a vertical line with an associated 
profile $S~=~C_1~A_2~C_3~C_4$: the first cluster $C_1$ represents the shading of an initial open segment 
of the vertical line, the atom $A_2$ represents the type of the upper endpoint of this segment, and the
clusters $C_3$ and $C_4$ represent the shadings of two adjacent open segments.
\begin{figure}[!!t]
\centering
\includegraphics[scale=0.85]{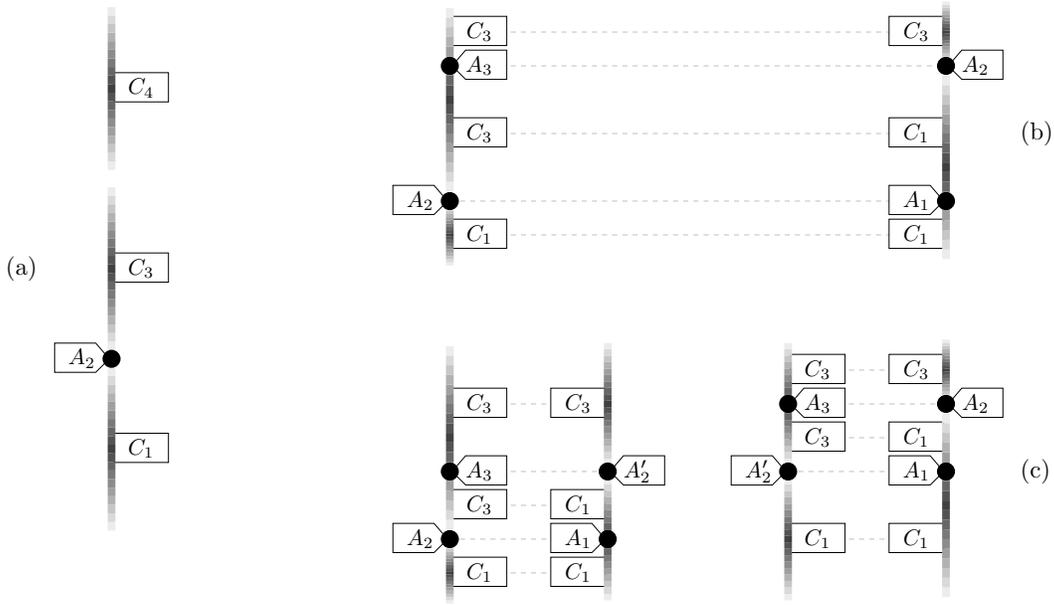}
\caption{Profiles (a), stripe expressions (b), and decompositions (c).}
\label{fig:profiles}
\end{figure}

To represent the types along the two vertical borders of a labelled stripe, we introduce
the notion of \emph{stripe expression}, which is a pair $E=(L,R)$ of \emph{left} and 
\emph{right profiles} having equal length and such that, for all $1\le i\le\len{E}$ ($=\len{L}=\len{R}$), 
$L(i)$ is an atom (resp., a cluster) if and only if $R(i)$ is an atom (resp., a cluster). 
We call any pair of the form $\big(L(i),R(i)\big)$, with $1\le i\le\len{E}$, a \emph{matched pair}. 

As an example, Figure \ref{fig:profiles}(b) depicts the left border and the right border of a labelled stripe, 
together with the associated stripe expression $E=(L,R)$, where $L~=~C_1~A_2~C_3~A_3~C_3$ and $R~=~C_1~A_1~C_1~A_2~C_3$. 

We say that an atom $A$ \emph{appears} in the left (resp., right) profile of a stripe expression 
$E=(L,R)$ if there is a position $1\le i\le\len{E}$ such that either $A=L(i)$ or $A\in L(i)$ 
(resp., either $A=R(i)$ or $A\in R(i)$) depending on whether $L(i)$ (resp., $R(i)$) is an atom or a cluster. 
By a slight abuse of notation, we denote by $\bigcup_{1\le i\le\len{E}}L(i)$ (resp., $\bigcup_{1\le i\le\len{E}}R(i)$) 
the set of all atoms that appear in the left (resp., right) profile of the stripe expression $E=(L,R)$. 

It is not difficult to see that
for every labelled stripe $\cP$, there exists a stripe expression $E$ whose left (resp., right) 
profile contains all and only the types of the points along the left (resp., right) border of $\cP$. 
For this, it suffices to consider the atoms that occur exactly once along each border -- we call
those atoms \emph{pivots} for short. The pivots will appear in the stripe expression $E$ and will be interleaved 
with the shadings of the segments that are intercepted at the coordinates of the pivots.
This shows how to construct a stripe expression $E$ that corresponds to a labelled stripe $\cP$.
Conversely, for some stripe expression $E$ there might exist no labelled stripe $\cP$ such that the shading 
of the left (resp., right) border of $\cP$ coincides with the set of all atoms appearing in the left 
(resp., right) profile of $E$. 
The reason is that the occurrences of atoms and clusters in $E$ might be inconsistent with the 
underlying requests and observables.
The rest of this section is devoted to overcome this problem, 
namely, to find suitable conditions under which a stripe expression admits a corresponding labelled stripe.
As a first step, we enforce suitable constraints on stripe expressions:

\begin{defi}\label{def:faithful}
We say that a stripe expression $E=(L,R)$ is \emph{faithful} if it satisfies the following properties:
\begin{description}   \item[(C1)] for all positions $1\le i<j\le\len{E}$, 
              we have $L(i) \,\CTstrict\, L(j)$ and $R(i) \,\CTstrict\, R(j)$;
  \medskip
  \item[(C2)] for all positions $1\le i\le\len{E}$, if $L(i)$ and $R(i)$ are clusters, 
              then we have $L(i) \,\CTstrict\, L(i)$ and $R(i) \,\CTstrict\, R(i)$;
  \medskip
  \item[(C3)] for all positions $1\le i\le j\le\len{E}$, 
              we have $L(i) \,\CURp\, R(j)$ and $L(j) \,\CLRp R(i)$;
  \medskip
  \item[(C4)] for all positions $1\le i\le\len{E}$, if $L(i)$ and $R(i)$ are atoms, 
              then we have 
\[\eqalign{
  \DTstrict\req(L(i)) &\subseteq
                              \bigcup_{j>i}\DTstrict\obs(L(j))\cr
  \DTstrict\req(R(i)) &\subseteq~ \bigcup_{j>i}\DTstrict\obs(R(j))}
  \enspace\eqalign{&,\cr&,}\enspace
  \eqalign{
  \DBstrict\req(L(i)) &\subseteq 
                              \bigcup_{j<i}\DBstrict\obs(L(j))\cr
  \DBstrict\req(R(i)) &\subseteq~ \bigcup_{j<i}\DBstrict\obs(R(j))\,;}
\]  \medskip
  \item[(C5)] for all positions $1\le i\le\len{E}$, if $L(i)$ and $R(i)$ are clusters, 
              then we have
\[\eqalign{
  \DTstrict\req(L(i)) &\subseteq 
                             \bigcup_{j\ge i}\DTstrict\obs(L(j))\cr
  \DTstrict\req(R(i)) &\subseteq~ \bigcup_{j\ge i}\DTstrict\obs(R(j))}
\enspace\eqalign{&,\cr&,}\enspace
  \eqalign{
  \DBstrict\req(L(i)) &\subseteq 
                              \bigcup_{j\le i}\DBstrict\obs(L(j))\cr
  \DBstrict\req(R(i)) &\subseteq~ \bigcup_{j\le i}\DBstrict\obs(R(j))\,.}
\]  \medskip
\end{description}
\end{defi}

\noindent Intuitively, the purpose of the first three conditions is to guarantee some {\sl consistency} constraints 
on the relationships between the requests and the observables of the atoms that appear in the left and right 
profiles of the given stripe expression, with the idea that the profiles represent the shadings of the two borders 
of a concrete labelled stripe. Similarly, the purpose of the last two conditions is to guarantee the {\sl fulfilment} 
of the existential requests of the left and right profiles along the two vertical directions $\CTstrict$ and $\CBstrict$.
From now on, we tacitly assume that every stripe expression is faithful (this can be easily checked).

\smallskip

Before enforcing further constraints on stripe expressions, we address a problem related to their representation.
First of all, we observe that a cluster that appears in a stripe expression may contain exponentially many atoms.
Thus, in principle, any explicit representation of a stripe expression may require exponential space. We cope with 
this problem by restricting to stripe expressions that are maximal with respect to a suitable partial order.
Formally, given two stripe expressions $E=(L,R)$ and $E'=(L',R')$, we write 
$E \unlhd E'$ 
(and read $E$ is \emph{dominated} by $E'$) if and only if
\begin{enumerate}[label=\roman*)]   \item $\len{E} = \len{E'}$;
  \item for all positions $1\le i\le\len{E}$, either $L(i)$, $R(i)$, $L'(i)$, and $R'(i)$ are atoms, 
        or $L(i)$, $R(i)$, $L'(i)$, and $R'(i)$ are clusters;
  \item for all positions $1\le i\le\len{E}$, either $L(i)=L'(i)$ and $R(i)=R'(i)$ hold, or $L(i)\subseteq L'(i)$ and $R(i)\subseteq R'(i)$ hold, 
        depending on whether $L(i)$, $R(i)$, $L'(i)$, and $R'(i)$ are atoms or clusters.
\end{enumerate}
As $\unlhd$ is a partial order, it makes sense to talk about \emph{maximal} (faithful) stripe expressions,
that is, stripe expressions which are not strictly dominated by other ones. 
The benefit of such a notion is that, given a cluster $C$ of a maximal stripe expression $E=(L,R)$, 
that is, $C=L(i)$ or $C=R(i)$ for some $1\le i\le\len{E}$, and  a generic atom $A$, one has
$$
  A\in C \qquad\text{if (and only if)}\qquad \begin{cases}
                                  \Darg{d}\req(A) ~=~\Darg{d}\req(C) \\[1ex]
                                  \Darg{d}\obs(A) ~\subseteq~ \Darg{d}\obs(C)
                                \end{cases}
                                \quad\text{for all spatial relations $\Carg{d}$}.
$$
It immediately follows that each cluster of a maximal stripe expression can be succinctly represented by listing 
all its requests and observables (recall that the number of requests and observables is at most 
linear in $\len{\varphi}$). 

In addition, one observes the following. If $E=(L,R)$ is a stripe expression and $1\le i<j\le\len{E}$ are the
positions of two {\sl different} matched pairs of clusters, that is, $\big(L(i),R(i)\big)\neq\big(L(j),R(j)\big)$,
then, due to the constraints of Definition \ref{def:faithful}, at least one of the following non-containments
is satisfied for some spatial relation $\Carg{d}\in\big\{\CTstrict,\CURp,\CULp\big\}$ and its inverse $\Carg{\inv{d}}$:
$$
\begin{array}{rclrcl}
  \Darg{d}\req(L(i)) &\,\supsetneq& \Darg{d}\req(L(j))  & \qquad\qquad 
  \Darg{d}\req(R(i)) &\,\supsetneq& \Darg{d}\req(R(j)) \\[2ex]
  \Darg{\inv{d}}\req(L(j)) &\,\supsetneq& \Darg{\inv{d}}\req(L(i))  & \qquad\qquad 
  \Darg{\inv{d}}\req(R(j)) &\,\supsetneq& \Darg{\inv{d}}\req(R(i))\ .
\end{array}
$$
It is worth noticing that $\Darg{d}\req(L(i)) \subseteq\! \Darg{d}\req(L(j))$ 
implies $\Darg{d}\req(L(i)) =\! \Darg{d}\req(L(j))$, and the same for the other conditions.
This means that any stripe expression can contain 
{\sl at most linearly many distinct matched pairs of clusters}.

From now on, we restrict ourselves to (faithful) {\sl maximal} stripe expressions that contain 
{\sl pairwise distinct matched pairs of clusters}. Thanks to this assumption and to the previous arguments, 
we can represent each stripe expression using space polynomial in $\len{\varphi}$.
Since the matched pairs of clusters in a stripe expression are pairwise distinct,
there are indeed at most linearly many such pairs in a stripe expression. Moreover, each matched 
pair of atoms is surrounded by two matched pairs of clusters. This implies that the length of
a stripe expression is at most linear in $\len{\varphi}$. Finally, as we argued earlier, 
each pair of atoms/clusters in a maximal stripe expression can be represented by
listing all the requests and observables in it, which are again linear in $\len{\varphi}$.

\subsection{Recursive decompositions of stripes}\label{subsec:decompositions}

Roughly speaking, Conditions C1--C5 of
Definition \ref{def:faithful} provide us with a guarantee that the natural spatial interpretation 
of a stripe expression $E$ is locally consistent with the view-to-type dependency. To enforce the 
global consistency and, in particular, to enforce the fulfilment of all existential requests, 
we need to introduce a suitable notion of decomposition. 
We start by dividing a given labelled stripe into a pair of thinner adjacent labelled sub-stripes; then,
we apply the decomposition recursively to every emerging sub-stripe. This yields an infinite 
tree-shaped decomposition of the initial structure, where each vertex of the tree represents a 
labelled (sub-)stripe and each edge represents a containment relationship between two labelled 
(sub-)stripes. 

\smallskip

To start with, we introduce a suitable equivalence relation between profiles. Intuitively, 
the equivalence relation identifies profiles that can be associated to the same vertical
line.

\begin{defi}\label{def:equivalence}
Two profiles $S$ and $S'$ are said to be \emph{equivalent} if 
\begin{itemize}
  \item the clusters that appear in $S$ and in $S'$ are the same;
  \item for each atom $S(i)$ that appears in $S$, either $S(i)$ also appears in $S'$ 
        or the two adjacent clusters $S(i-1)$ and $S(i+1)$ coincide and they both
        contain the atom $S(i)$, and symmetrically for each atom $S'(i)$ of $S$.
\end{itemize}
\end{defi}

\noindent As an example, two profiles of the form $S~=~C_1~C_2~C_2$ and $S'~=~C_1~A_1~C_1~C_2$,
with $A_1\in C_1$, are equivalent; on the contrary, the profile $S$ is not equivalent 
to any profile $S''~=~C_1~A_1~C_2~C_2$, unless $A_1\in C_1$ and $C_1=C_2$.

\smallskip

Decompositions of stripe expressions are defined as follows. 

\begin{defi}\label{def:decompositions}
 Let $E=(L,R)$ be a stripe expression. A \emph{decomposition} of $E$ is any pair of 
stripe expressions $(E_1,E_2)$, with $E_1=(L_1,R_1)$ and $E_2=(L_2,R_2)$, that satisfies
the following matching conditions:
\begin{description}   \item[(M1)] $L_1$ and $L$ are equivalent,
  \medskip
  \item[(M2)] $R_2$ and $R$ are equivalent,
  \medskip
  \item[(M3)] $R_1$ and $L_2$ are equivalent.
  \medskip
\end{description}
\end{defi}

We say that a matched pair $\big(L(i),R(i)\big)$ of the stripe expression $E$ \emph{corresponds}
to a matched pair $\big(L_1(i_1),R_1(i_1)\big)$ of the left stripe expression $E_1$ under the 
decomposition $(E_1,E_2)$ of $E$ if there is a position $1\le i_2\le\len{E_2}$ such that $L(i)\equalin L_1(i_1)$, 
$R(i)\equalin R_2(i_2)$, and $R_1(i_1)\equalin L_2(i_2)$ hold, where $\equalin$ denotes either 
the identity relation $=$ between atoms or between clusters, or the membership relation $\in$ 
between atoms and clusters, or the inverse membership relation $\ni$ between clusters and atoms.
A symmetric definition can be given for correspondences with matched pairs of the right stripe 
expression $E_2$.

As an example, Figure \ref{fig:profiles}(c) depicts a decomposition of the stripe expression $E=(L,R)$, 
where $L~=~C_1~A_2~C_3~A_3~C_3$ and $R~=~C_1~A_1~C_1~A_2~C_3$. Note that, under such a decomposition, 
the matched pair $(C_3,C_1)$ of $E$ corresponds to the three matched pairs $(C_3,C_1)$, $(A_3,A'_2)$, 
and $(C_3,C_3)$ of $E_1$ and to the three matched pairs $(C_1,C_1)$, $(A'_2,A_1)$, and $(C_3,C_1)$
of $E_2$.

\smallskip

By iteratively applying decompositions, starting from an initial stripe expression, one 
obtains an infinite tree-shaped structure, called decomposition tree:

\begin{defi}\label{def:decompositiontree}
A \emph{decomposition tree} is an infinite complete binary labelled 
tree $\cT=\ang{V,E,\prj{1}{},\prj{2}{}}$, where
\begin{itemize}
  \item $V$ is the set of vertices;
  \item $\prj{1}{}$ and $\prj{2}{}$ are the left and right successor relations;
  \item $E$ is a labelling function that associates with each vertex $v\in V$ a stripe expression $E(v)$
        in such a way that the pair $\big(E(\prj{1}{v}),\,E(\prj{2}{v})\big)$ is a decomposition of 
        the stripe expression $E(v)$.
\end{itemize}
\end{defi}

\noindent Hereafter, we fix a decomposition tree $\cT=\ang{V,E,\prj{1}{},\prj{2}{}}$. Given a vertex $v$ in $\cT$
and the associated stripe expression $E(v)=(L,R)$, we shortly denote by $E(v)[L]$ (resp., $E(v)[R]$) 
its left profile $L$ (resp., its right profile $R$). 

We observe that, due to the matching conditions M1--M3, 
if $v$ and $v'$ are two vertices of a decomposition tree $\cT=\ang{V,E,\prj{1}{},\prj{2}{}}$ and $v'$ is
right-adjacent to $v$ (possibly without being a sibling), then the right profile $E(v)[R]$ of $v$ and 
the left profile $E(v')[L]$ of $v'$ are equivalent. Note that this is also consistent with the spatial 
interpretation of stripe expressions that imposes the right profile of $v$ and the left profile of $v'$
to represent the same vertical line.

\smallskip

We now enforce suitable conditions on the decomposition tree $\cT$ in order to guarantee that every
existential request of every atom that appears in a stripe expression $E(v)$ is eventually fulfilled 
by an observable of an atom in a (possibly different) stripe expression $E(v')$.  
Recall that, thanks to Conditions C4--C5 of
Definition \ref{def:faithful}, all requests along the directions $\CTstrict$ and $\CBstrict$
are fulfilled within the same stripe expression $E(v)$. It thus remains to consider the requests 
along the directions $\CULp$, $\CURp$, $\CLLp$, $\CLRp$. In the following, we consider a generic 
vertex $v$ of $\cT$ and we look at the {\sl right}-oriented requests of the atoms/clusters that 
appear in the {\sl left} profile $E(v)[L]$; symmetrically, we look at the {\sl left}-oriented 
requests for the atoms/clusters that appear in the {\sl right} profile $E(v)[R]$. 
For the sake of brevity, we only provide the fulfilment conditions for the requests of the 
left profile $E(v)[L]$ along the direction $\CURp$ (the reader can easily devise the correct 
definitions for the remaining directions):

\begin{defi}\label{def:fulfillment}
Let $v$ be a vertex of the decomposition tree $\cT$ and let $\alpha$ be a formula in $\closure(\varphi)$. 
We say that $\alpha$ is \emph{locally fulfilled as a $\DURp$-request at vertex $v$} 
if for all positions $1\le i\le\len{E(v)}$, at least one of the following conditions holds:
\begin{description}   \item[(F1)] $\alpha\nin\DURp\req\big(E(v)[L](i)\big)$;
  \medskip
  \item[(F2)] $\alpha\in\DURp\req\big(E(v)[R](i)\big)$;
  \medskip
  \item[(F3)] $\alpha\in\DURp\obs\big(E(v)[R](j)\big)$ for some position $i\le j\le\len{E(v)}$;
  \medskip
  \item[(F4)] there exist two positions $1\le i_1\le j_1\le\len{E(\prj{1}{v})}$ such that
        \begin{enumerate}[label=\roman*)]           
        \item the matched pair $\big(E(v)[L](i),\,E(v)[R](i)\big)$ corresponds
              to the matched pair $\big(E(\prj{1}{v})[L](i_1),\,E(\prj{1}{v})[R](i_1)\big)$ 
                under the decomposition $\big(E(\prj{1}{v}),\,E(\prj{2}{v})\big)$ of $E(v)$,
          \item $\alpha\in\DURp\obs\big(E(\prj{1}{v})[R](j_1)\big)$.
        \end{enumerate}
  \medskip
\end{description}
\end{defi}

\noindent We are now able to express the conditions that make a fulfilled decomposition tree
a valid representation of some concrete labelled stripe:

\begin{defi}\label{def:globallyfulfilled}
A decomposition tree $\cT$ is \emph{globally fulfilled} if it satisfies the following conditions:
\begin{description}   \item[(G1)] if $v_0$ is the root of $\cT$, for all spatial relations
        $\Carg{d}\in\big\{\CURp,\CLRp\big\}$ (resp., $\Carg{d}\in\big\{\CULp,\CLLp\big\}$)
        and all positions $1\le i\le\len{E(v_0)}$, the set $\Darg{d}\req\big(E(v_0)[R](i)\big)$ 
        (resp., $\Darg{d}\req\big(E(v_0)[L](i)\big)$) is empty;
  \medskip
  \item[(G2)] for every formula $\alpha\in\closure(\varphi)$, every spatial relation $\Carg{d}$, 
        and every infinite path $\pi$ in $\cT$, there exist infinitely many vertices $v$ along $\pi$ 
        such that $\alpha$ is locally fulfilled as a $\Darg{d}$-request at vertex $v$.
  \medskip
\end{description}
\end{defi}

\noindent Finally, we say that a globally fulfilled decomposition tree $\cT$
\emph{satisfies} $\varphi$ if it contains a ($\varphi$-)atom $A$ such that $\varphi\in A$.

\section{A tree pseudo-model property}\label{sec:pseudomodel}

In this section, we establish a tree pseudo-model property for 
satisfiable formulas of Cone Logic. 
We first show that, given any labelled stripe $\cP=\ang{X\times\bbQ,(R_a)_{a\in\Sigma}}$
-- e.g., a model of $\varphi$ -- there is a globally fulfilled decomposition tree $\cT$ 
whose stripe expressions {\sl contain} at least the types of the points of $\cP$ 
(Theorem \ref{prop:completeness}).
Then, we prove that, given a globally fulfilled decomposition tree $\cT$, there is a labelled stripe $\cP=\ang{X\times\bbQ,(R_a)_{a\in\Sigma}}$ whose shading {\sl coincides} with the set of all atoms that appear in the stripe expressions 
of $\cT$ (Theorem \ref{prop:soundness}).
The two results together provide us with a way to represent \emph{over-approximations} of shadings 
of labelled stripes by means of globally fulfilled decompositions trees (formally, an over-approximation of a stripe is a set of types that contains the shading of that stripe). 

In Section \ref{sec:solution} we shall see how the correspondence between labelled stripes and 
globally fulfilled decompositions trees allows us to reduce the satisfiability problem for a 
formula $\varphi$ of Cone Logic to the problem of deciding the existence of a globally fulfilled decomposition tree that satisfies $\varphi$.

\begin{thm}[completeness]\label{prop:completeness}
For every labelled stripe $\cP=\ang{X\times\bbQ,(R_a)_{a\in\Sigma}}$, there is a globally fulfilled 
decomposition tree $\cT=\ang{V,E,\prj{1}{},\prj{2}{}}$ such that 
$$
  \type_{\cP}\big(X\times\bbQ\big) 
  \quad\subseteq
  \bigcup\limits_{\begin{smallmatrix}v\in V\\1\le i\le\len{E(v)}\end{smallmatrix}}\Bigl(~E(v)[L](i) ~\cup~ E(v)[R](i)~\Bigr).
$$
\end{thm}

\begin{proof}
Let $\cP=\ang{X\times\bbQ,(R_a)_{a\in\Sigma}}$ be a labelled stripe, where $X$ is a closed interval of the rational numbers, and let $T=\ang{V,\prj{1}{},\prj{2}{}}$ be the infinite, complete, and unlabelled  binary tree. We need to associate with each vertex $v$ of $T$ a suitable stripe expression $E(v)$. 
To do that, we recursively divide the labelled stripe $\cP$ into substripes, each one corresponding to some vertex 
$v$ of $T$; then, we collect the types of the points along the borders of the emerging (sub)stripes and accordingly construct the stripe expressions. 
There is, however, a little complication in this construction, due to the fact that the resulting decomposition tree must be globally fulfilled and it must contain all the types of the points in $\cP$. 
To enforce these conditions, we need to choose properly the $x$-coordinates along which we 
divide the labelled (sub)stripe associated with each vertex $v$. 

Before turning to the main construction, we give some preliminary definitions. We fix, once and for all, an enumeration $\theta:\bbN\then X$ of the rational numbers in the closed interval $X$ (recall that the set $X$ is countable). 
Moreover, we define the \emph{parity} of a vertex $v$ in $T$ to be the distance from the root modulo $1+4\cdot\len{\closure(\varphi)}$. The parity value $0$ will play a special role, while
the parity values from $1$ to $4\cdot\len{\closure(\varphi)}$ are identified with triples of the form $(\lambda,\alpha,\Carg{d})$, 
where $\lambda\in\{L,R\}$, $\alpha\in\closure(\varphi)$, and either $\Carg{d}\in\big\{\CURp,\CLRp\big\}$ or $\Carg{d}\in\big\{\CULp,\CLLp\big\}$
depending on whether $\lambda=L$ or $\lambda=R$. By a slight abuse of terminology, we say that a vertex $v$ has parity $0$ or $(\lambda,\alpha,\Carg{d})$.

\smallskip
\noindent
{\bf The construction of the decomposition tree.} \
We start by associating with each vertex $v$ of $T$ (i) a stripe $[x_v^L,x_v^R]\times\bbQ$, with $x_v^L,x_v^R\in X$,
and (ii) a stripe expression $E_v$ whose left and right profiles contain, respectively, the types of the 
points along the left border $P_v^L=\{x_v^L\}\times\bbQ$ and the right border $P_v^R=\{x_v^R\}\times\bbQ$.
In doing that, we shall guarantee that if $(\lambda,\alpha,\Carg{d})$ is the parity of the vertex $v$, then $\alpha$ 
is locally fulfilled as a $\Darg{d}$-request at the vertex $v$ (intuitively, this gives a fair policy for the fulfilment 
of all requests at all vertices). We give such definitions by exploiting an induction on the distance of the vertex 
$v$ from the root. If $v_0$ is the root of $T$, then we simply let $x_{v_0}^L=\min(X)$ and $x_{v_0}^R=\max(X)$. 
Consider now a generic vertex $v$ in $T$ and suppose, by inductive hypothesis,
that the two coordinates $x_v^L$ and $x_v^R$ have been defined. We consider the types of the points along the left 
border $P_v^L=\{x_v^L\}\times\bbQ$ and along the right border $P_v^R=\{x_v^R\}\times\bbQ$ of the corresponding stripe $[x_v^L,x_v^R]\times\bbQ$, and we introduce an equivalence relation $\sim_v$ over $\bbQ$ such that $y\sim_v y'$ if and only if, for all spatial relations $\Carg{d}\in\big\{\CTstrict,\CBstrict,\CULp,\CURp,\CLLp,\CLRp\big\}$, we have:
$$
\begin{array}{rclrcl}
  \Darg{d}\req(x_v^L,y) &\,=& \Darg{d}\req(x_v^L,y') & \quad\quad\quad  \Darg{d}\req(x_v^R,y) &\,=& \Darg{d}\req(x_v^R,y') \\[2ex]
  \Darg{d}\obs(x_v^L,y) &\,=& \Darg{d}\obs(x_v^L,y') & \quad\quad\quad  \Darg{d}\obs(x_v^R,y) &\,=& \Darg{d}\obs(x_v^R,y')
\end{array}
$$
(for the sake of brevity, we denote by $\Darg{d}\req(x,y)$ and $\Darg{d}\obs(x,y)$, respectively, 
the set of $\Darg{d}$-requests and the set of $\Darg{d}$-observables of the type of the point $p=(x,y)$). 

It can be easily checked (e.g., by exploiting view-to-type dependency) that the equivalence relation $\sim_v$ 
has finite index and it induces a partition of $\bbQ$ into some subsets $Y_{v,1}<...<Y_{v,k_v}$ (here we write
$Y<Y'$ as a shorthand for $y<y'$ for all $y\in Y$ and all $y'\in Y'$). 
Then, we refine the partition into a finite sequence of convex sets $Y'_{v,1}<...<Y'_{v,h_v}$, with $h_v\ge k_v$, 
that are either {\sl singletons} or {\sl open intervals}. 
Accordingly, we divide the left border $P_v^L$ (resp., the right border $P_v^R$) into a sequence of (singleton or open) 
segments $P_{v,i}^L=\bigsettc{(x_v^L,y)}{y\in Y'_{v,i}}$ (resp., $P_{v,i}^R=\bigsettc{(x_v^R,y)}{y\in Y'_{v,i}}$), with 
$1\le i\le h_v$.
On the basis of the partition $P_{v,1}^L,...,P_{v,h_v}^L$ of $P_v^L$ and the partition $P_{v,1}^R,...,P_{v,h_v}^R$ 
of $P_v^R$, we define a (possibly non-maximal) stripe expression $E_v=(L_v,R_v)$ of length $\len{E_v} = h_v$ by 
specifying the components $L_v(i)$ and $R_v(i)$ of each matched pair. Let $1\le i\le h_v$ be a position of $E_v$. 
If both segments $P_{v,i}^L$ and $P_{v,i}^R$ are singletons of the form $\{p_{v,i}^L\}$ and $\{p_{v,i}^R\}$, 
respectively, then we let $L_v(i)$ be the atom $\type_\cP(p_{v,i}^L)$ and $R_v(i)$ be the atom $\type_\cP(p_{v,i}^R)$. 
Otherwise, if $P_{v,i}^L$ and $P_{v,i}^R$ are open segments, then we let $L_v(i)$ be the cluster 
$\type_\cP(P_{v,i}^L)$ and $R_v(i)$ be the cluster $\type_\cP(P_{v,i}^L)$. 

We observe that the above-defined stripe expression $E_v$ is not maximal with respect to the partial order $\unlhd$
introduced in Subsection \ref{subsec:profiles}. As stripe expressions of decomposition trees are required to be 
maximal, we cannot directly label $v$ with $E_v$ in our decomposition tree. However, if the stripe expression 
$E_v$ is known to be faithful, then we can label $v$ with a maximal (faithful) stripe expression $E(v)$ that 
dominates $E_v$. Unfortunately, it is not clear from the above constructions if the stripe expression 
$E_v$ is faithful.
We shall prove that this is actually the case later.
For the moment, the reader can simply assume that the stripe expression $E(v)$ associated with 
vertex $v$ is undefined when $E_v$ is not faithful.

It remains to specify the coordinate $x_v^M$ along which we divide the current stripe $[x_v^L,x_v^R]\times\bbQ$.
We choose such a coordinate $x_v^M$ by looking at the parity of the vertex $v$. Precisely, if $v$ has parity $0$, 
then we define $x_v^M$ to be the first coordinate, according to the order given by the fixed enumeration $\theta$ 
of $X$, that is strictly between $x_v^L$ and $x_v^R$. Intuitively, this choice will guarantee that every coordinate 
$x\in X$ is eventually identified with either $x_u^L$ or $x_u^R$, for some vertex $u$ in $T$. Otherwise, if $v$ has 
parity $(\lambda,\alpha,\Carg{d})$, then we let $I$ be the set of all positions $1\le i\le\len{E(v)}$ such that 
$\alpha\in\Darg{d}\req\big(E(v)[\lambda](i)\big)$, $\alpha\nin\Darg{d}\req\big(E(v)[\lambda'](i)\big)$, and
$\alpha\nin\Darg{d}\obs\big(E(v)[\lambda'](i)\big)$, where $\lambda'$ is either $R$ or $L$ depending on whether 
$\lambda=L$ or $\lambda=R$. Depending on whether $\Carg{d}$ is downward-oriented or upward-oriented 
(i.e., whether $\Carg{d}\in\big\{\CLRp,\CLLp\big\}$ or $\Carg{d}\in\big\{\CURp,\CULp\big\}$), we let $i$ 
be either the least or the greatest position in $I$ (if $I$ is empty, then the choice of the coordinate $x_v^M$ 
is irrelevant, provided that it is strictly between $x_v^L$ and $x_v^R$). We then choose arbitrarily a point 
$p\in P_{v,i}^\lambda$ and a point $q$ such that $p \Carg{d} q$ and $\alpha\in\Darg{d}\obs\big(\type_\cP(q)\big)$ and 
we force 
$x_v^M$ to be the $x$-coordinate of $q$. 
Note that since $\alpha$ is neither in $\Darg{d}\req\big(E(v)[\lambda'](i)\big)$ nor in $\Darg{d}\obs\big(E(v)[\lambda'](i)\big)$, 
the coordinate $x_v^M$ is strictly between $x_v^L$ and $x_v^R$. 
Accordingly, if $v_1=\prj{1}{v}$ and $v_2=\prj{2}{v}$ are the left and right successors of the vertex $v$ in $T$, 
then we let $x_{v_1}^L=x_v^L$, $x_{v_2}^R=x_v^R$, and $x_{v_1}^R=x_{v_2}^L=x_v^M$. Finally, we inductively apply 
the above construction to the successors $v_1$ and $v_2$ of $v$. 

It is worth pointing out that the stripe expression $E_v$ is decomposed into a left stripe expression 
$E_{v_1}$ and a right stripe expression $E_{v_2}$ in such a way that the matching conditions M1--M3
of Definition \ref{def:decompositions} 
are satisfied. Given that $v$ has parity $(\lambda,\alpha,\Carg{d})$, it can be easily checked that 
the formula $\alpha$ is locally fulfilled as a $\Darg{d}$-request at vertex $v$. 
Analogous properties hold also when we replace each stripe expression $E_v$ 
with the maximal dominating one $E(v)$. What remains to be shown is that
\begin{enumerate}[label=\roman*)]   \item all stripe expressions $E_v$ are faithful (possibly non-maximal),
  \item all types of points of the labelled stripe $\cP$ appear as atoms in some stripe expression $E_v$ 
        (hence they also appear in the maximal stripe expression that dominates $E_v$), 
  \item the decomposition tree $\cT=\ang{V,E,\prj{1}{},\prj{2}{}}$, obtained from $T$ by labelling each 
        vertex $v$ with the maximal stripe expression $E(v)$ that dominates $E_v$, is globally fulfilled.
\end{enumerate}

\smallskip
\noindent
{\bf All stripe expressions are faithful.} \
We fix a vertex $v$ of $\cT$ and we prove that the stripe expression $E_v$ satisfies 
Conditions C1--C5 of Definition \ref{def:faithful}. 
We do this by exploiting the view-to-type dependency and the fact that the atoms (resp., clusters) in the two
profiles $L_v$ and $R_v$ of $E_v$ arise from the types (resp., shadings) of the singleton (resp., open) segments $P_{v,i}^L$ and $P_{v,i}^R$. 
As for Condition C1, we consider two atoms $A$ and $B$ that appear along the same profile of $E_v$ at positions $i$ and $j$, respectively, with $1\le i<j\le\len{E_v}$. 
Let $A=L_v(i)$ and $B=L_v(j)$ (the cases where $L_v(i)$ and/or $L_v(j)$ are clusters or $A$ and $B$ lie along the right profile $R_v$ are similar and thus omitted). 
By construction,
the corresponding segments $P_{v,i}^L$ and $P_{v,j}^L$ are singletons whose points $p\in P_{v,i}^L$ 
and $q\in P_{v,j}^L$ satisfy $p \CTstrict q$. From the view-to-type dependency, we conclude that $\type_\cP(p) \,\CTstrict\, \type_\cP(q)$, whence $L_v(i) \,\CTstrict\, L_v(j)$. Similar arguments can be used to prove Conditions C2 and C3. 
As for the last two conditions, we consider a request $\alpha$ of an atom $L_v(i)$ along the direction $\CTstrict$ 
(the cases of requests of atoms/clusters of left/right profiles along directions $\CTstrict$ and $\CBstrict$ are all similar).
By construction, 
the segment $P_{v,i}^L$ consists of a single point $p$. Moreover, since $\DTstrict\alpha\in\type_\cP(p)$, there is a point $q$ such that $p \,\CTstrict\, q$ and $\alpha\in\type_\cP(q)$. Again by construction, 
there is a segment $P_{v,j}^L$, with $j>i$, that contains the point $q$. We thus conclude that 
$\alpha$ is an observable of $L_v(j)$ along the direction $\CTstrict$.

\smallskip
\noindent
{\bf All types appear in stripe expressions.} \
Let $p=(x,y)$ be a geitemizeneric point in the labelled stripe $\cP$ and 
let $\pi$ be the infinite path of the infinite binary tree $T$ such that 
$x\in[x_v^L,x_v^R]$ for all vertices $v$ along $\pi$ (note that such an 
infinite path $\pi$ exists since $x$ belongs to the first interval 
$X=[x_{v_0}^L,x_{v_0}^R]$ associated with the root $v_0$ and 
$[x_v^L,x_v^R] = [x_{\prj{\!1\;}{\!v\!}}^L,x_{\prj{\!1\;}{\!v\!}}^R] \cup 
[x_{\prj{\!2\;}{\!v\!}}^L,x_{\prj{\!2\;}{\!v\!}}^R]$ for all vertices $v$).
Since $x\in X$ and $\theta$ is an enumeration of $X$, 
there is a natural number $n$ such that $\theta(n)=x$. 
Moreover, since $\pi$ contains infinitely many vertices with parity 
$0$, there must be one such vertex $v$ satisfying $x=x_v^L$, $x=x_v^R$, or $x=x_v^M$ 
($=x_{\prj{\!1\;}{\!v\!}}^R=x_{\prj{\!2\;}{\!v\!}}^L$). 
Hence, the type of the point $p$ appears as an atom in one of the stripe 
expressions $E_v$, $E_{\prj{\!1\;}{\!v\!}}$, or $E_{\prj{\!2\;}{\!v\!}}$ 
that are associated with the vertex $v$, its left-successor $\prj{1}{v}$, or 
its right-successor $\prj{2}{v}$. 

\smallskip
\noindent
{\bf The decomposition tree is globally fulfilled.} \
We conclude by showing that the decomposition tree $\cT$, that results from $T$ by labelling each vertex $v$ 
with the maximal stripe expression $E(v)$ that dominates $E_v$, is globally fulfilled. By construction, the 
root $v_0$ of $\cT$ satisfies $\Darg{d}\req\big(E(v_0)[R](i)\big)=\emptyset$ (resp., $\Darg{d}\req\big(E(v_0)[L](i)\big)=\emptyset$) 
for all positions $1\le i\le\len{E(v_0)}$ and all spatial relations $\Carg{d}\in\big\{\CURp,\CLRp\big\}$ 
(resp., $\Carg{d}\in\big\{\CULp,\CLLp\big\}$). This proves Condition G1 of Definition 
\ref{def:globallyfulfilled}. As for Condition G2, we consider a formula $\alpha\in\closure(\varphi)$, 
a spatial relation $\Carg{d}\in\big\{\CULp,\CURp,\CLLp,\CLRp\big\}$, and an infinite path $\pi$ in $\cT$. 
We let $\lambda$ be either $L$ or $R$ depending on whether $\Carg{d}$ is right-oriented or left-oriented. 
For every $n\in\bbN$, we can find a vertex $v_n$ along $\pi$ that is at distance at least $n$ from the root 
and that has parity exactly $(\lambda,\alpha,\Carg{d})$. Thus, we know from the previous arguments that 
there exist infinitely many vertices $v$ along $\pi$ where $\alpha$ is locally fulfilled as a $\Darg{d}$-request. 
This shows that $\cT$ is a globally fulfilled decomposition tree. 
\end{proof}

\begin{thm}[soundness]\label{prop:soundness}
For every globally fulfilled decomposition tree $\cT=\ang{V,E,\prj{1}{},\prj{2}{}}$, there is a 
labelled stripe $\cP=\ang{X\times\bbQ,(R_a)_{a\in\Sigma}}$ such that
$$
  \type_{\cP}\big(X\times\bbQ\big) 
  \quad=
  \bigcup\limits_{\begin{smallmatrix}v\in V\\1\le i\le\len{E(v)}\end{smallmatrix}}\Bigl(~E(v)[L](i) ~\cup~ E(v)[R](i)~\Bigr).
$$
\end{thm}

\begin{proof}
Let $\cT=\ang{V,E,\prj{1}{},\prj{2}{}}$ be a globally fulfilled decomposition tree. As a first step, we associate 
with each vertex $v$ of $\cT$ two coordinates $x_v^L,x_v^R\in\bbQ$ as follows. 
If $v$ is the root of $\cT$, then we let $x_v^L=0$ and $x_v^R=1$. If $v$ is a vertex of $\cT$ and $v_1=\prj{1}{v}$ 
and $v_2=\prj{2}{v}$ are its left and right successors, then, assuming that both values $x_v^L$ and $x_v^R$ are 
defined, we let $x_{v_1}^L=x_v^L$, $x_{v_2}^R=x_v^R$, and $x_{v_1}^R=x_{v_2}^L=\frac{x_v^L+x_v^R}{2}$. 
We collect all these values into a set $X\subseteq\bbQ$:
$$
  X ~=~ \bigsettc{x_v^L}{v\in V} \cup \bigsettc{x_v^R}{v\in V} ~=~ \bigsettc{\tfrac{i}{2^n}}{i,n\in\bbN,~0\le i\le 2^n}.
$$ 
Note that $X$ is strictly included in the interval $[0,1]$ of $\bbQ$ and 
it has minimum and maximum elements. However, since all countable dense 
linear orders with minimal and maximal elements are isomorphic, we can give $X$ the status 
of a closed interval of the rational numbers. By the same abuse of terminology, we call 
the structure $X\times\bbQ$ a stripe and, for any $x<x'\in X$, we denote by $[x,x']$ the 
set of all points $x''\in X$ such that $x\le x''\le x'$. 

The next step consists of dividing the left and right borders of each (sub)stripe $[x_v^L,x_v^R]\times\bbQ$ vertically
on the basis of the stripe expression $E(v)$ and the matching relations with the successor stripe expressions. 
For technical reasons, we will make use of the subset of dyadic rationals to mark the 
endpoints of some vertical segments.
A \emph{dyadic rational} is a rational number of the form $\tfrac{i}{2^n}$, 
for some $i\in\bbZ$ and $n\in\bbN$. It can be easily checked that
dyadic rationals are densely interleaved with non-dyadic ones. 
We will associate with each vertex $v$ of $\cT$ and each position $1\le i\le\len{E(v)}$ a convex subset $Y_{v,i}$ 
of $\bbQ$ in such a way that the following conditions are satisfied:
\begin{enumerate}
  \item $Y_{v,1}<Y_{v,2}<...<Y_{v,\len{E(v)}}$;
  \item $\bigcup_{1\le i\le\len{E(v)}}Y_{v,i}=\bbQ$;
  \item if $E(v)[L](i)$ and $E(v)[R](i)$ are atoms, then $Y_{v,i}$ is a singleton whose unique element is a dyadic rational;
  \item if $E(v)[L](i)$ and $E(v)[R](i)$ are clusters, then $Y_{v,i}$ is an open interval of rational numbers;
  \item for all vertices $v$, all successors $v'$ of $v$, and all positions $1\le i\le\len{E(v)}$ and $1\le i'\le\len{E(v')}$,
        $Y_{v',i'}$ intersects $Y_{v,i}$ if and only if the $i$-th matched pair of $E(v)$ corresponds to the $i'$-th matched pair 	of $E(v')$ 
        under the decomposition induced by $\cT$.
\end{enumerate}
The above sets $Y_{v,i}$ can be built by exploiting a simple induction based on the 
breadth-first traversal of the vertices of $\cT$. We omit the formal construction of the sets $Y_{v,i}$, which is tedious and not interesting, and we only remark that, in order to enforce the above properties, one needs to exploit the density of 
dyadic and non-dyadic rational numbers. 
During the inductive steps that define the sets $Y_{v,i}$, we can enforce an additional invariant, 
that will be explained a few paragraphs below and that will only be used towards the end of the proof.

Let $L$ be the maximal length of a stripe expression. We fix, once and for all, an enumeration 
$$
  \theta:~\bbN\then\big(\{\emptyset\}~\cup~\bbQ\times\bbQ\big)^L
$$
of all possible $L$-tuples $(\tilde{Y}_1,...,\tilde{Y}_L)$ of (possibly empty) {\sl closed intervals} of $\bbQ$ (the reason 
for considering closed intervals, instead of generic ones, is that there are uncountably many open intervals in $\bbQ$).

Let us focus on the induction step during which the sets $Y_{v,1}$, $...$, $Y_{v,\len{E(v)}}$ are associated 
with a certain vertex $v$.
We say that a tuple $\theta(m)=(\tilde{Y}_1,...,\tilde{Y}_L)$ is \emph{compatible with the decomposition at vertex $v$} 
if, given the choices of the sets $Y_{v',1}$, $...$, $Y_{v',\len{E(v')}}$ for all vertices $v'$ that precede $v$ in 
the breadth-first traversal of $\cT$, it is possible to choose the sets $Y_{v,1},...,Y_{v,\len{E(v)}}$ without violating 
the above constraints and in such a way that the containments $Y_{v,i}\supseteq\tilde{Y}_i$ are satisfied for all positions 
$1\le i\le\len{E(v)}$. In order to properly choose the sets $Y_{v,i}$, we mark the vertex $v$ with the {\sl first} 
natural number $m_v$ such that (i) $\theta(m_v)$ is compatible with the decomposition at vertex $v$ and (ii)
$m_v$ does not already mark a proper ancestor $v'$ of $v$ such that $E(v')=E(v)$ (note that such a number $m_v$ 
exists and is unique). The number $m_v$ is called the \emph{fingerprint} of $v$. 
The sets $Y_{v,i}$  are chosen in such a way that they satisfy conditions (1)$-$(5) above and the following additional invariant: 

\medskip\noindent
{\bf Additional invariant.}
{\it If $m_v$ is the fingerprint of $v$ and $\theta(m_v)=(\tilde{Y}_1,...\tilde{Y}_L)$, then $Y_{v,i}\supseteq\tilde{Y}_i$ for all $1\le i\le\len{E(v)}$.}

\medskip
Now, 
we associate with every vertex $v$ and position $1\le i\le\len{E(v)}$ the two (singleton or open) vertical segments $P_{v,i}^L=\{x_v^L\}\times Y_{v,i}$ and $P_{v,i}^R=\{x_v^R\}\times Y_{v,i}$.
Clearly, the union of these segments cover the entire stripe $X\times\bbQ$:
$$
  \bigcup_{\begin{smallmatrix} v\in\cT \\ 1\le i\le\len{E(v)}\end{smallmatrix}} \!\!P_{v,i}^L \cup P_{v,i}^R ~~=~~ X\times\bbQ.
$$
The last step of the construction consists of defining a labelling $(R_a)_{a\in\Sigma}$ of the stripe $X\times\bbQ$ whose induced shading coincides with the set of all atoms of the stripe expressions of $\cT$. 
To this end, for each letter $a\in\Sigma$ and point $p\in X\times\bbQ$, 
we specify whether or not $p$ belongs to the subregion $R_a$. 
We first consider those points $p$ that belong to one or more {\sl singleton} segments $P_{v,i}^\lambda$, with 
$\lambda\in\{L,R\}$ (we call these points \emph{primary}). Given a primary point $p$, we choose arbitrarily some 
vertex $v$ of $\cT$, some position $1\le i\le\len{E(v)}$, and some direction $\lambda\in\{L,R\}$ such that 
$P_{v,i}^\lambda=\{p\}$.
$E(v)[\lambda](i)$ is necessarily an atom, and we accordingly let $p\in R_a$ if and only if the proposition letter $a$ occurs positively in $E(v)[\lambda](i)$. This defines the labelling of primary points.
To specify the labelling of those points that are only covered by open segments (\emph{secondary} points),  a slightly more complex construction is needed, which is based on the notion of ``shuffle''. 
More precisely, for each non-empty set $C\subseteq\sP(\Sigma)$, we fix a function $\eta_C:\bbQ\then C$ such that 
for all $y, y'\in\bbQ$, with $y<y'$, and all $c\in C$, there is a {\sl non-dyadic} rational $y''\in\bbQ$ satisfying $y<y''<y'$ 
and $\eta(y'')=c$ (we call this function the \emph{shuffle} of $C$). A crucial feature of the notion of 
shuffle is that if one removes some  (possibly all) dyadic rationals from the labelled linear order $\eta_C$, he obtains 
a labelled linear order which is isomorphic to $\eta_C$ itself, and, symmetrically, if one inserts some isolated positions 
in $\eta_C$ labelled by elements of $C$, he obtains again a labelling isomorphic to $\eta_C$. 
Now, for each secondary point $p=(x,y)$, we choose arbitrarily some vertex $v$ of $\cT$, some position $1\le i\le\len{E(v)}$, and some direction $\lambda\in\{L,R\}$ such that $p\in P_{v,i}^\lambda$.
$E(v)[\lambda](i)$ is a cluster, and we accordingly let $p\in R_a$ 
if and only f $a\in \eta_C(y)$, where $C=\bigsettc{A\cap\Sigma}{A\in E(v)[\lambda](i)}$. 

In view of the above definitions, one may think of the set of proposition letters 
associated with a certain point $p$ as dependent on the particular choice of the arguments $v,i$, and $\lambda$ such that $p\in P_{v,i}^\lambda$. This is actually not the case. To prove it, one can exploit the matching conditions M1--M3
of Definition \ref{def:decompositions} 
and a simple induction to verify the following claims:
\begin{enumerate}[label=\roman*)]   \item if two singleton segments $P_{v,i}^\lambda$ and $P_{v',i'}^{\lambda'}$ cover the same (primary) point $p$, then we have
        $E(v)[\lambda](i) \,=\, E(v')[\lambda'](i')$ and hence $p$ is labelled by $a$ if and only if $a\in E(v)[\lambda](i)$, if and 	only if $a\in E(v')[\lambda'](i')$;
  \item if a singleton segment $P_{v,i}^\lambda$ and an open segment $P_{v',i'}^{\lambda'}$ cover the same (primary) 		point $p$, then $E(v)[\lambda](i) \,\in\, E(v')[\lambda'](i')$ and hence there is an atom $A\in E(v')[\lambda'](i')$ that 		contains exactly the labels of the point $p$ and possibly other more complex subformulas;  
  \item if $P_{v,i}^\lambda$ and $P_{v',i'}^{\lambda'}$ are two overlapping open segments, then 
  	$E(v)[\lambda](i) \,=\, E(v')[\lambda'](i')$ and thus the labelling of the secondary points in 
  	$P_{v,i}^\lambda \cap P_{v',i'}^{\lambda'}$ (naturally ordered from bottom to top) is isomorphic 
  	to the shuffle $\eta_C$, with $C=\bigsettc{A\cap\Sigma}{A\in E(v)[\lambda](i)}$
        or, equivalently, $C=\bigsettc{A\cap\Sigma}{A\in E(v')[\lambda'](i')}$;
  \item if $P_{v,i}^\lambda$ is an open segment, then the primary points inside $P_{v,i}^\lambda$ have dyadic $y$-			coordinates and thus they must be interleaved by secondary points; together with the previous claim, this
  	implies that the labelling of $P_{v,i}^\lambda$ is isomorphic to the shuffle $\eta_C$, with 
  	$C=\bigsettc{A\cap\Sigma}{A\in E(v)[\lambda](i)}$.
\end{enumerate}
What remains to do is to show that the shading of the labelled stripe $\cP=\ang{X\times\bbQ,(R_a)_{a\in\Sigma}}$ 
coincides with the set of atoms that appear in the stripe expressions of the decomposition tree $\cT$. We prove this
by an induction based on increasing sets of formulas closed under 
subformulas, that is, we consider sets $F$ that contain all subformulas $\beta$ of $\alpha$ whenever $\alpha\in F$. 
The rest of the proof is devoted to show the following statement 
(for $F=\closure(\varphi)$, it leads to the desired conclusion).

\begin{clm}
Let $F$ be a set of formulas closed under 
subformulas. For all vertices $v\in V$, all
positions $1\le i\le\len{E(v)}$, all directions $\lambda\in\{L,R\}$, and all open intervals $Y\subseteq\bbQ$ 
such that $P_{v,i}^\lambda \cap (X\times Y)\neq\emptyset$,
\begin{itemize}
  \item if $A=E(v)[\lambda](i)$ is an atom, then the unique point $p\in P_{v,i}^\lambda \cap (X\times Y)$ 
        satisfies $\type_\cP(p)\cap F=A\cap F$;
  \item if $E(v)[\lambda](i)$ is a cluster, then, for every atom $A\in E(v)[\lambda](i)$, there is
        a point $p\in P_{v,i}^\lambda \cap (X\times Y)$ (and, vice versa, for every point $p\in P_{v,i}^\lambda$,
        there is an atom $A\in E(v)[\lambda](i)$) such that $\type_\cP(p)\cap F=A\cap F$.
\end{itemize}
\end{clm}

\medskip\noindent
We fix a set $F$ of formulas closed under 
subformulas, a vertex $v\in V$, a position $1\le i\le\len{E(v)}$, 
and a direction $\lambda=L$ (the case $\lambda=R$ is symmetric) such that $E(v)[\lambda](i)$ is a cluster 
(the case where $E(v)[\lambda](i)$ is an atom is similar). We also fix an open interval $Y\subseteq\bbQ$ 
such that $P_{v,i}^\lambda\cap(X\times Y)\neq\emptyset$. We prove the above claim by exploiting an induction 
on the size of 
$F$.
\begin{enumerate}[label=\roman*)]   \item Base case: $F=\Sigma$. This case is trivial as, from previous arguments, we know that the labelling of the 
        open segment $P_{v,i}^\lambda$ is isomorphic to the shuffle $\eta_C$, where 
        $C=\bigsettc{A\cap\Sigma}{A\in E(v)[\lambda](i)}$.

  \item Inductive case: $F=F'\uplus\settc{\neg\beta}{\beta\in F'}\uplus\settc{\beta_1\vel\beta_2}{\beta_1,\beta_2\in F'}$,
        where $F'$ is a set of formulas closed under 
	subformulas. As in the previous case, the claim trivially follows 
        from the inductive hypothesis on $F'$ and from the definition of atom.

  \item Inductive case: $F=F'\uplus\settc{\DTstrict\beta}{\beta\in F'}$, where $F'$ is a set of formulas closed under 
	subformulas (the case $F=F'\uplus\settc{\DBstrict\beta}{\beta\in F'}$ is symmetric).
        We fix an atom $A\in E(v)[\lambda](i)$ and we prove that there is a point $p\in P_{v,i}^\lambda\cap(X\times Y)$
        satisfying $\type_\cP(p)\cap F=A\cap F$ (using similar arguments one can show that, for every point 
        $p\in P_{v,i}^\lambda$, there is an atom $A\in E(v)[\lambda](i)$ satisfying $\type_\cP(p)\cap F=A\cap F$). 
        
        We start by observing that, thanks to the inductive hypothesis on $F'$, there is a point 
        $p\in P_{v,i}^\lambda\cap(X\times Y)$ such that $\type_\cP(p)\cap F'=A\cap F'$. It is now sufficient to 
        show that $\alpha\in\type_\cP(p)$ if and only if $\alpha\in A$ for all formulas $\alpha\in F\setminus F'$.
        
        Let $\alpha=\DTstrict\beta\in F\setminus F'$, with $\beta\in F'$, and suppose that $\alpha\in\type_\cP(p)$. 
        By definition of type, there is a point $q$ such that $p \,\CTstrict\, q$ and $\beta\in\type_\cP(q)$.
        Let $j$ ($\ge i$) be the unique position of $E(v)$ such that $q\in P_{v,j}^\lambda$. By the inductive 
        hypothesis, 
	there is an atom $B$ that either coincides with $E(v)[\lambda](j)$ or belongs to 
        $E(v)[\lambda](j)$, depending on whether $E(v)[\lambda](j)$ is an atom or a cluster, and that contains the 
        subformula $\beta$. Moreover, by Conditions C1 and C2 of Definition \ref{def:faithful}, 
	$A \,\CTstrict\, B$. Since $\beta\in B$, we obtain $\beta\in\DTstrict\req(A)$ and thus $\alpha\in A$. 

        As for the converse implication, suppose that $\alpha=\DTstrict\beta\in A$. Clearly, $\beta\in\DTstrict\req(A)$. 
        Moreover, by Condition C4, there must be a position $j\ge i$ of $E(v)$ and an atom $B$
        that either coincides with $E(v)[L](j)$ or belongs to $E(v)[L](j)$ and that satisfies $\beta\in\DTstrict\obs(B)$. 
        Let $Y'=\bigsettc{y\in\bbQ}{p \,\CTstrict\, (x_v^\lambda,y)}$. We observe that $X\times Y'$ is an open 
        vertical segment that intersects $P_{v,j}^\lambda$. By applying the inductive hypothesis to the vertex $v$, the 
        position $j$, the atom $B$, and the open interval $Y'$, we derive the existence of a point 
        $q\in P_{v,j}^\lambda\cap(X\times Y')$ such that $\beta\in\type_\cP(q)$. Finally, since $p \,\CTstrict\, q$, 
        we conclude that $\beta\in\DTstrict\req(p)$ and thus $\alpha\in\type_\cP(p)$.

  \item Inductive case: $F=F'\uplus\settc{\DURp\beta}{\beta\in F'}$, where $F'$ is a set of formulas closed under 
	subformulas (the cases for the remaining operators $\DULp$, $\DLRp$, $\DLLp$ can be dealt with using
        similar arguments). This is 
	the most interesting and complex case, as it puts together all the pieces of the puzzle that we have introduced so 		
	far, e.g., Definitions \ref{def:dependencies}, \ref{def:faithful}, \ref{def:fulfillment}, and \ref{def:globallyfulfilled}.
        As in the previous case, we fix an atom $A\in E(v)[L](i)$ and we prove that there is a point 
        $p\in P_{v,i}^L\cap(X\times Y)$ satisfying $\type_\cP(p)\cap F=A\cap F$ (the proof of the 
        converse direction, that fixes a point and obtains an atom, is similar). The new ingredient here
        is that we will consider {\sl multiple candidate points} obtained from the inductive hypothesis. 
        Precisely, we partition the open interval $Y_{v,i} \cap Y$ into an infinite sequence $Y'_1>Y'_2>Y'_3>...$ 
        of smaller open intervals (this is possible because the subordering $Y_{v,i} \cap Y$ is isomorphic to $\bbQ$). 
        For each of these intervals $Y'_l$, we apply the inductive hypothesis on $F'$ and we obtain a point 
        $p'_l\in P_{v,i}^L\cap(X\times Y'_l)$ such that, for all $\beta\in F'$, $\beta\in\type_\cP(p'_l)$ 
        if and only if $\beta\in A$. 
	All points $p'_1,p'_2,p'_3,...$ lie along the same open vertical segment 
        $P_{v,i}$, they are naturally ordered from top to bottom, and they get arbitrarily close to the 
        lower endpoint of the segment $P_{v,i}$ (symmetric arrangements of points should be considered
        for the downward-oriented operators $\DLLp$ and $\DLRp$).
        Below, we prove that $\type_\cP(p'_l)\cap F\subseteq A\cap F$ holds for all points $p'_1,p'_2,p'_3,...$. 
        Later on, we will prove that the converse containment holds for all but finitely many such points.

        \smallskip
        Let $\alpha=\DURp\beta\in F\setminus F'$, with $\beta\in F'$, and suppose that $\alpha\in\type_\cP(p'_l)$. 
        By definition of type, there is a point $q$ such that $p'_l \,\CURp\, q$ and $\beta\in\type_\cP(q)$.
        Starting from $v$, we define an ascending sequence of vertices $v_0,...,v_n$, 
        where $v_0=v$, $v_n$ is the root of $\cT$,
        and $v_{k+1}$ is the parent of $v_k$ for all $0\le k<n$. Given $0\le k\le n$, we denote by $i_k$ the unique position 
        of $E(v_k)$ such that the interval $Y_{v_k,i_k}$ contains the $y$-coordinate of the point $p'_l$ (note that $i_0=i$). 
        Clearly, any two intervals $Y_{v_k,i_k}$ and $Y_{v_{k+1},i_{k+1}}$ have non-empty intersection. 
        Therefore, thanks to the constraints enforced at the beginning of the proof, either 
        $E(v_k)[L](i_k) \,\equalin\, E(v_{k+1})[L](i_{k+1})$ or $E(v_k)[R](i_k) \,\equalin\, E(v_{k+1})[R](i_{k+1})$,
        depending on whether $v_k$ is the left successor or the right successor of $v_{k+1}$. 
        From Condition C3 of Definition \ref{def:faithful}, we also know that
        $E(v_k)[L](i_k) \,\CURp\, E(v_k)[R](i_k)$. Putting all together and exploiting the transitivity
        of the relation $\CURp$ over atoms/clusters, we obtain $E(v)[L](i) \,\CURp\, E(v_n)[R](i_n)$.

        By using a similar technique, we define an infinite descending sequence of vertices 
        $v_n,v_{n+1},...$ in such a way that the point $q$ lies always inside the stripe 
        $[x_{v_{k+1}}^L,x_{v_{k+1}}^R]\times\bbQ$, but never along the left border.
        As before, we denote by $i_k$ the unique position of $E(v_k)$ such that the interval $Y_{v_k,i_k}$ 
        contains the $y$-coordinate of the point $p'_l$. This guarantees that either 
        $E(v_k)[L](i_k) \,\equalin\, E(v_{k+1})[L](i_{k+1})$ or $E(v_k)[R](i_k) \,\equalin\, E(v_{k+1})[R](i_{k+1})$ 
        holds, depending on whether $v_{k+1}$ is the left successor or the right successor of $v_k$. 
        We know from Condition C3 that $E(v_k)[L](i_k) \,\CURp\, E(v_k)[R](i_k)$ and hence,
        using again transitivity, $E(v)[L](i) \,\CURp\, E(v_k)[R](i_k)$ for all $k\ge n$. 

        Consider the first vertex $v_k$ in the sequence $v_n,v_{n+1},...$ such that $q$ lies along the right 
        border of the corresponding stripe expression $E(v_k)$, namely, $q\in P_{v_k}^R$ (the existence of such a vertex 
        follows from the definition of the stripe $X\times\bbQ$). 
        Let $j$ be the (unique) position of $E(v_k)$ such that $Y_{v_k,j}$ contains the $y$-coordinate of $q$.
        Clearly, we have $j\ge i_k$. Moreover, from the inductive hypothesis, 
        we have $\beta\in\DURp\obs\big(E(v_k)[R](j)\big)$
        and hence, by Conditions C1 and C2, 
        $E(v_k)[R](i_k) \,\CTstrict\, E(v_k)[R](j)$. 
        Finally, we exploit Definition \ref{def:dependencies}, and in particular the fact that 
        $E(v)[L](i) \,\CURp\, E(v_k)[R](i_k) \,\CTstrict\, E(v_k)[R](j)$ implies $E(v)[L](i) \,\CURp\, E(v_k)[R](j)$,
        to conclude that $\beta\in\DURp\req\big(E(v)[L](i)\big)$. As $\alpha=\DURp\beta$ and $A\in E(v)[L](i)$, 
        this shows that $\alpha\in A$.

        \smallskip
        We now prove that the converse containment $A\cap F\subseteq\type_\cP(p'_l)\cap F$ holds for at least one 
        of the infinitely many points $p'_1,p'_2,p'_3,...$. Let $\alpha=\DURp\beta\in A$ ($\in E(v)[L](i)$). 
        For the sake of brevity, we denote by $Y'$ the set of the $y$-coordinates of all points 
        $p'_1,p'_2,p'_3,...$ (note that $Y'\subseteq Y_{v,i}$). Moreover, given a vertex $v'$ in $\cT$ and a position 
        $1\le i'\le\len{E(v')}$, we say that $i'$ is an \emph{interesting} position of $v'$ if the interval $Y_{v',i'}$ 
        contains infinitely many coordinates from the set $Y'$. 
        Note that every vertex $v'$ has at least one interesting position (this follows from simple counting 
        arguments, since the infinite set $Y'$ is partitioned into finitely many sets of the form $Y_{v',i'} \cap Y'$, 
        with $1\le i'\le\len{E(v')}$). It is also easy to see that there is at most one interesting position $i'$ 
        for each vertex $v'$ (this follows from the fact the set $Y'$ has a unique accumulation point in the 
        completion $\bbR\cup\{-\infty,+\infty\}$ of $\bbQ$).

        Now, we consider the ascending sequence of vertices $v_0,...,v_n$ that starts from $v$ and reaches
        the root of $\cT$, where each $v_{k+1}$ is the parent of $v_k$, for $0\le k<n$.
        Let $i_0,...,i_n$ be the interesting positions of the vertices $v_0,...,v_n$, respectively. 
        By exploiting a simple induction on $k$, we prove that $\beta\in\DURp\req\big(E(v_k)[L](i_k)\big)$ 
        for all $0\le k\le n$. For $k=0$ the claim follows easily since $i_0=i$ is the interesting position 
        of the vertex $v_0=v$. For the inductive step, we assume that the claim holds for $k$ and we prove it
        for $k+1$. We distinguish two cases depending on whether $v_k$ is the left successor or the right
        successor of $v_{k+1}$. In this first case, since the two intervals $Y_{v_k,i_k}$ and $Y_{v_{k+1},i_{k+1}}$ 
        overlap, we know that $E(v_k)[L](i_k) \,\equalin\, E(v_{k+1})[L](i_{k+1})$ and hence we immediately obtain
        $\beta\in\DURp\req\big(E(v_{k+1})[L](i_{k+1})\big)$. In the second case, we consider the left sibling $v'$ 
        of $v_k$ and its interesting position $i'$. As the two intervals $Y_{v_k,i_k}$ and $Y_{v',i'}$ overlap,
        we have $E(v')[R](i') \,\equalin\, E(v_k)[L](i_k)$, whence $\beta\in\DURp\req\big(E(v')[R](i')\big)$. 
        Moreover, Condition C3 implies $E(v')[L](i') \,\CURp\, E(v')[R](i')$, whence 
        $\beta\in\DURp\req\big(E(v')[L](i')\big)$. Finally, as the two intervals $Y_{v',i'}$ and $Y_{v_{k+1},i_{k+1}}$ 
        overlap and $v'$ is the left-successor of $v_{k+1}$, we have $E(v')[L](i') \,\equalin\, E(v_{k+1})[L](i_{k+1})$,
        whence $\beta\in\DURp\req\big(E(v_{k+1})[L](i_{k+1})\big)$.

        Below, we use a similar technique to build an infinite descending sequence of vertices $v_n,v_{n+1}, \ldots$ 
        such that, for all $k\ge n$, the interesting position $i_k$ of $v_k$ satisfies both 
        $\beta\in\DURp\req\big(E(v_k)[L](i_k)\big)$ and $\beta\nin\DURp\req\big(E(v_k)[R](i_k)\big)$.
        As for the base case ($k=n$), it suffices to recall that $v_n$ is the root of $\cT$ 
        and that Condition G1 of Definition \ref{def:globallyfulfilled} implies 
        $\beta\nin\DURp\req\big(E(v_n)[R](i_n)\big)$. 
        As for the inductive step, we assume that $v_k$ is defined and that $i_k$ is its interesting position,
        and we define $v_{k+1}$ as follows. Let $w_1$ and $w_2$ be, respectively, the left 
	and the right successor of $v_k$, and let $j_1$ and $j_2$ be the interesting positions of 
	$w_1$ and $w_2$, respectively.
        Since the intervals $Y_{v_k,i_k}$, $Y_{w_1,j_1}$, and $Y_{w_2,j_2}$ are pairwise 
        overlapping, it holds that
        $E(v_k)[L](i_k) \,\equalin\, E(w_1)[L](j_1)$, $E(v_k)[R](i_k) \,\equalin\, E(w_2)[R](j_2)$, and
        $E(w_1)[R](j_1) \,\equalin\, E(w_2)[L](j_2)$. This implies that $\beta\in\DURp\req\big(E(w_1)[L](j_1)\big)$, 
        $\beta\nin\DURp\req\big(E(w_2)[R](j_2)\big)$ and either $\beta\nin\DURp\req\big(E(w_1)[R](j_1)\big)$ or
        $\beta\in\DURp\req\big(E(w_2)[L](j_2)\big)$. Depending on the latter two cases, we define $v_{k+1}$
        to be either $w_1$ or $w_2$; accordingly, the interesting position $i_{k+1}$ of $v_{k+1}$ is either 
        $j_1$ or $j_2$. 

        Let us consider now the above-defined infinite path  
	$\pi=v_n,v_{n+1}, \ldots$
        From Condition G2 of Definition \ref{def:globallyfulfilled}, we know that $\pi$ contains 
        infinitely many vertices $v_k$ where the formula $\beta$ is locally fulfilled as a $\DURp$-request. 
        By construction, all points $p'_1,p'_2,p'_3,...$ lie either strictly to the left of each stripe 		
        $[x_{v_k}^L,x_{v_k}^R]\times\bbQ$ or along its left border $P_{v_k}^L$. Moreover, since 
        $\beta\in\DURp\req\big(E(v_k)[L](i_k)\big)$ and 
        $\beta\nin$ $\DURp\req\big(E(v_k)[R](i_k)\big)$, we know that, among the $4$ cases envisaged by Definition 
        \ref{def:fulfillment}, only the last two cases (Condition F3 and F4) 
        can be satisfied by each vertex $v_k$ and its interesting 
        position $i_k$. We thus distinguish between two subcases.

        \smallskip\noindent
        {\bf Subcase F3.}\
        If $\pi$ contains a vertex $v_k$ that satisfies Condition F3, then
        we have $\beta\in\DURp\obs\big(E(v_k)[R](j)\big)$ for some position $j$ that is greater than or equal to
        the interesting position $i_k$ of $v_k$. By 
	the inductive hypothesis, 
	there exists a point
        $q\in P_{v_k,j}^R$ such that $\beta\in\DURp\obs\big(\type_\cP(q)\big)$. 
        Moreover, since $j\ge i_k$ and $Y_{v_k,i_k}$
        contains infinitely many elements from $Y'$, we have that the elements of $Y_{v_k,j}$ are strictly greater than all but 
        finitely many elements of $Y'$. In particular, we have that $p'_l \,\CURp\, q$ for all but finitely many points $p'_l$.
        This allows us to conclude that $\alpha=\DURp\beta\in\type_\cP(p'_l)$ for all 
        but finitely many point $p'_l$ in $P_{v,i}^L\cap(X\times Y)$.

        \smallskip\noindent
        {\bf Subcase F4.}\
        If $\pi$ contains infinitely many vertices $v_{k_1},v_{k_2},...$ satisfying Condition F4,
	then, for each index $k_h$, with $h\ge 1$, 
        the stripe expression $E(\prj{1}{v_{k_h}})$ contains two positions $i^\downarrow_h\le j^\downarrow_h$ such that 
        (i) the $i^\downarrow_h$-th matched pair of $E(\prj{1}{v_{k_h}})$ corresponds to the $i_h$-th matched pair of 
        $E(v_{k_h})$, where $i_h$ is the interesting position of $v_{k_h}$, and
        (ii) $\beta\in\DURp\obs\big(E(\prj{1}{v_{k_h}})[R](j^\downarrow_h)\big)$.
        Without loss of generality (e.g., by restricting to a suitable subsequence of vertices), 
        we can assume that all stripe expressions $E(\prj{1}{v_{k_1}})$, $E(\prj{1}{v_{k_2}})$, ... 
        coincide, and hence we can denote them simply by $E^\downarrow$. Similarly, we can assume that all 
        indices $i^\downarrow_1,i^\downarrow_2,...$ (resp., $j^\downarrow_1,j^\downarrow_2,...$) coincide, 
        and hence we can denote them simply by $i^\downarrow$ (resp., $j^\downarrow$).
	Consider now the tuples $(\tilde{Y}_1,...,\tilde{Y}_L)$ of closed intervals of $\bbQ$ such that 
        $\tilde{Y}_{i^\downarrow}$ contains infinitely many $y$-coordinates from the set $Y'$ and 
        $\tilde{Y}_j=\emptyset$ for all other indices $j\in\{1,...,L\}\setminus\{i^\downarrow\}$. 
        We call these tuples $(\tilde{Y}_1,...,\tilde{Y}_L)$ \emph{interesting tuples} and we let $M$
        be the set of indices of all interesting tuples, according to the fixed enumeration $\theta$
        that we introduced at the beginning of the proof.
        We observe that there are infinitely many interesting tuples that are {\sl compatible} with 
        the decompositions at the vertices $\prj{1}{v_{k_1}}$, $\prj{1}{v_{k_2}}$, .... In particular,
        this means that infinitely many indices from $M$ appear as {\sl fingerprints} of vertices along 
        $\pi$ that might be different from $\prj{1}{v_{k_1}}$, $\prj{1}{v_{k_2}}$, ..., but whose stripe  
        expressions coincide with $E^\downarrow$. Let $v^\downarrow$ be any of these vertices.  
        From the construction given at the beginning of this proof, it follows that 
        $Y_{v^\downarrow,i^\downarrow}\supseteq\tilde{Y}_{i^\downarrow}$. 
        In particular, as $\tilde{Y}_{i^\downarrow}$ contains infinitely many $y$-coordinates from the set $Y'$, 
        we have that $i^\downarrow$ is the interesting position of $v^\downarrow$. 
        Since $E(v^\downarrow)=E^\downarrow=E(\prj{1}{v_{k_h}})$ for all $h\ge 1$, and 
        $\beta\in\DURp\obs\big(E^\downarrow[R](j^\downarrow)\big)$ for some $j^\downarrow\ge i^\downarrow$,
        by 
	the inductive hypothesis, 
	there exists a point $q\in P_{v^\downarrow,j^\downarrow}^R$ 
        such that $\beta\in\type_\cP(q)$. To conclude, it suffices to observe that the elements of 				
        $Y_{v^\downarrow,j^\downarrow}$ are greater than all but finitely many elements of $Y'$. 
        This shows that $p'_l \,\CURp\, q$, and thus $\alpha=\DURp\beta\in\type_\cP(p'_l)$ for all 
        but finitely points $p'_l$ in $P_{v,i}^L\cap(X\times Y)$.
\end{enumerate}
This concludes the proof. 
\end{proof}

 \section{Reducing Cone Logic to a fragment of CTL}\label{sec:solution}

In this section, we make use of the tree pseudo-model property of Cone Logic to devise a decision procedure 
for its satisfiability problem. More precisely, thanks to the results shown in Section \ref{sec:pseudomodel}, the problem 
of establishing whether 
a formula $\varphi$ of Cone Logic is satisfiable over the labelled rational plane is reducible 
to the problem of checking the existence of a globally fulfilled decomposition tree $\cT$ that satisfies $\varphi$. 
The effectiveness of such an approach stems from the fact that the properties that characterize 
a globally fulfilled decomposition tree can be expressed in (a proper fragment of) CTL. 
This allows us to immediately reduce the satisfiability problem for Cone Logic to that for CTL, which is known to be in \exptime \cite{ctl,complexity_of_ctl_fragments}. 
From a practical point of view, this is already an interesting result, since there exist a number 
of efficient decision procedures for CTL. However, we will improve it by showing that 
the satisfiability problem for Cone Logic is in \pspace. This is done by further reducing the 
satisfiability problem for the fragment of CTL that captures Cone Logic to the universality 
problem for symbolic representations of non-deterministic B\"uchi automata. In the next 
section, we will see that the \pspace upper bound is actually tight.
 
\begin{thm}\label{th:pspace}
The satisfiability problem for Cone Logic, 
over the class of all labelled rational planes as well as
over the class of all labelled (rational or real) planes, is in \pspace.
\end{thm}

\begin{proof}
To start with, we recall that in Section \ref{sec:logic} (Remark \ref{rem:rationalplanes} and Remark \ref{rem:stripes})
we show that the satisfiability problem for Cone Logic, interpreted over the class of labelled rational planes 
(and, similarly, 
over the class of labelled, rational or real, planes) is reducible to the 
same problem over the class of labelled rational stripes. 
In the following, we first show how to reduce this latter problem to the satisfiability problem for a suitable fragment of CTL
(this theorem), and then to the universality problem for symbolically represented non-deterministic B\"uchi automata (next section).

The first step of the proof consists of translating, in polynomial time, a given  formula $\varphi$ of Cone Logic
into an equi-satisfiable conjunction $\tilde{\varphi}$ of CTL formulas of the forms: 
$$
  \lambda, \qquad\quad \AG\lambda, \qquad\quad \AG\EX\lambda, \qquad\quad \AG\delta, \qquad \mbox{ or } \quad \AG\AF\delta,
$$
where $\lambda$ and $\delta$
respectively denote a plain propositional formula and 
a CTL formula that uses the modality $\AX\!$ (only in a {\sl positive} way) and 
no other modality. 
Let us call the above conjuncts \emph{basic} CTL formulas. 

In the following, we show how to encode a 
decomposition tree $\cT$ by means of an infinite binary tree $\tilde{\cT}$ with labels only on vertices.
Such an encoding is needed because CTL formulas are not able to distinguish the two successor relations 
of a binary tree. 
First, we introduce three fresh proposition letters $0$, $1$, $2$ and we encode the two successor relations $\prj{1}{}$ 
and $\prj{2}{}$ of $\cT$ by giving each vertex $v$ either label $0$, $1$, or $2$, depending on whether $v$ is the 
root, $v=\prj{1}{u}$, or $v=\prj{2}{u}$, where $u$ is the parent of $v$. 
The resulting tree can be logically defined (up to bisimulation) using a suitable conjunction of basic CTL formulas 
over the signature $\{0,1,2\}$:
$$
  \tilde{\varphi}_{\tree} ~=~ (0 \et \neg 1 \et \neg 2) \,\et \AG\AX\bigl(\neg 0 \et \neg(1\et2)\bigr) \,\et \AG\EX 1 \,\et \AG\EX 2
$$
The next step consists of the encoding of the stripe expressions of $\cT$ by means of an additional labelling 
which is defined on top of the previous one. 
Since the number of atoms/clusters can be exponential in $\len{\varphi}$, we need to encode {\sl one by one} the subformulas of each atom/cluster that occur in each position of a given profile. 
To do this, we denote by $N$ the maximal length of a stripe expression (recall that $N$ is linear in $\len{\varphi}$ 
under the assumption that stripe expressions contain pairwise distinct matched pairs of clusters). For each index 
$1\le i\le N$, each formula $\alpha\in\closure(\varphi)$, and each spatial relation 
$\Carg{d}\in\big\{\CTstrict,\CBstrict,\CULp,\CURp,\CLLp,\CLRp\big\}$, we introduce eight fresh proposition letters: 
$$
\begin{array}{cccc}
  L^{\atom}_i, &\qquad\quad L^{\cluster}_i, &\qquad\quad L^{\Dargsmall{d}\obs}_{i,\alpha}, &\qquad\quad L^{\Dargsmall{d}\req}_{i,\alpha}, \\
  R^{\atom}_i, &\qquad\quad R^{\cluster}_i, &\qquad\quad R^{\Dargsmall{d}\obs}_{i,\alpha}, &\qquad\quad R^{\Dargsmall{d}\req}_{i,\alpha}. 
\end{array}
$$
Intuitively, 
$L^{\atom}_i$ (resp., $L^{\cluster}_i$) holds at a vertex $v$ of $\tilde{\cT}$ if and only if the position 
$i$ of the left profile $E(v)[L]$ of $v$ in $\cT$ contains an atom (resp., a cluster). Similarly, 
$L^{\Dargsmall{d}\obs}_{i,\alpha}$ (resp., $L^{\Dargsmall{d}\req}_{i,\alpha}$) holds at a vertex $v$ of $\tilde{\cT}$ 
if and only if the subformula $\alpha$ belongs to the set of observables $\Darg{d}\obs\bigl(E(v)[L](i)\bigr)$ (resp., to the set of requests $\Darg{d}\req\bigl(E(v)[L](i)\bigr)$). Analogous rules are used to encode the right profiles $E(v)[R]$. 
Note that, since we restrict ourselves to maximal stripe expressions, the above encoding uniquely determines the 
matched pairs of the stripe expressions in $\cT$. 

We now show how to enforce the various sanity conditions on the encoding of $\cT$. 
Conditions C1--C5 of Definition \ref{def:faithful} can be easily encoded by means of a basic 
CTL formula $\AG\lambda_{\text{C1--C5}}$ that holds over the encoding of $\cT$, 
where $\lambda_{\text{C1--C5}}$ is a propositional formula of size polynomial in $\len{\varphi}$. 
Enforcing the matching conditions M1--M3 of Definition
\ref{def:decompositions} requires some additional work. 
For this, it is convenient to explicitly write down the correspondence relationships between the matched pairs of 
a vertex $v$ and the matched pairs of its successors $\prj{1}{v}$ and $\prj{2}{v}$. 
For each triple of indices $i, i_1, i_2$, with $1\le i,i_1,i_2\le N$, we 
introduce a fresh proposition letter $M_{i,i_1,i_2}$ such that
$M_{i,i_1,i_2}$ holds at a vertex $v$ of the encoding of $\cT$ if and only if $E(v)[L](i)\equalin E(\prj{1}{v})[L](i_1)$, 
$E(v)[R](i)\equalin E(\prj{2}{v})[R](i_2)$, and $E(\prj{1}{v})[R](i_1)\equalin E(\prj{2}{v})[L](i_2)$ hold 
over the decomposition tree $\cT$. 
Using a basic CTL formula $\AG\delta$, where $\delta$ contains only positive occurrences of modality $\AX\!$ and no occurrence of other modalities, and it has size polynomial in $\len{\varphi}$, one can check the consistency of proposition letters $M_{i,i_1,i_2}$ at each vertex $v$ with the labellings that define the stripe expressions $E(v)$, $E(\prj{1}{v})$, and $E(\prj{2}{v})$. 
Moreover, enforcing the matching conditions M1--M3 amounts to 
checking the following three simple properties on each vertex $v$ of $\cT$:
\begin{enumerate}[label=\roman*)]   \item for all $1\le i\le\len{E(v)}$, $M_{i,i_1,i_2}$ holds at $v$ for some $1\le i_1\le\len{E(\prj{1}{v})}$ and some $1\le i_2\le\len{E(\prj{2}{v})}$, 
  \item for all $1\le i_1\le\len{E(\prj{1}{v})}$, $M_{i,i_1,i_2}$ holds at $v$ for some $1\le i\le\len{E(v)}$ and some $1\le i_2\le\len{E(\prj{2}{v})}$,
  \item for all $1\le i_2\le\len{E(\prj{2}{v})}$, $M_{i,i_1,i_2}$ holds at $v$ for some $1\le i\le\len{E(v)}$ and some $1\le i_1\le\len{E(\prj{1}{v})}$.
\end{enumerate}
The above properties are clearly expressible by a propositional formula $\lambda_{\text{M1--M3}}$ of small size.
As for the property of global fulfilment (see Definition \ref{def:globallyfulfilled}), we can enforce 
Condition G1 by a simple propositional formula $\lambda_{\text{G1}}$ evaluated at the
root of the tree, and Condition G2 by a conjunction of basic formulas of the form $\AG\AF\delta_{i,\alpha}^{\Cargsmall{d}}$, where $\delta_{i,\alpha}^{\Cargsmall{d}}$ contains only positive occurrences of modality $\AX\!$ and no occurrence of other modalities, $i$ ranges over $\{1,...,N\}$, $\alpha$ ranges over $\closure(\varphi)$, and $\Carg{d}$ ranges over $\big\{\CTstrict,\CBstrict,\CULp,\CURp,\CLLp,\CLRp\big\}$. 
It remains to check the existence of an atom $A$ in $\cT$ such that $\varphi\in A$. 
Without loss of generality, we can assume that the formula $\varphi$ starts with a modality among $\DULp$, $\DURp$, $\DLLp$, and $\DLRp$. This guarantees that $\varphi$ appears at some vertex $v$ of $\cT$ if and only if it appears at its root. Under such an assumption, a simple propositional formula $\lambda_\varphi$ evaluated at the root of the tree can enforce the existence of an atom/cluster of a stripe expression of $\cT$ that contains $\varphi$. Let $\tilde{\varphi}_{\path}$ be the  conjunction of the above-defined basic CTL formulas:
$$
\begin{array}{rcl}
  \tilde{\varphi}_{\path} &=&     \AG\lambda_{\text{C1--C5}} \,\et \AG\delta \,\et \AG\lambda_{\text{M1--M3}}                             
                                \,\et\, \lambda_{\text{G1}} \,\et \!\bigwedge\limits_{i,\alpha,\Cargsmall{d}}\!\!\!\AG\AF\delta_{i,			
                                \alpha}^{\Cargsmall{d}}  \;\,\et\; \lambda_\varphi
\end{array}
$$ 
We can conclude that any formula $\varphi$ of Cone Logic can be translated into a CTL formula $\tilde{\varphi}~=~\tilde{\varphi}_{\tree} \et\tilde{\varphi}_{\path}$, where both $\tilde{\varphi}_{\tree}$ and $\tilde{\varphi}_{\path}$ are conjunctions of basic CTL formulas.
Moreover, $\varphi$ occurs in some globally fulfilled decomposition tree $\cT$, that witnesses $\varphi$ at its root, if and only if $\tilde{\varphi}$ is satisfiable.

\smallskip
In order to complete the proof, we show how to obtain a \pspace decision procedure to check the satisfiability of the CTL formula $\tilde{\varphi}$.
The first conjunct $\tilde{\varphi}_{\tree}$ 
defines a $\{0,1,2\}$-labelled tree, where each vertex 
has at least two successors, distinguished by means of the labels $1$ and $2$. 
We denote such a 
tree by $T$ (up to bisimulation there is only one such structure). 
The second conjunct $\tilde{\varphi}_{\path}$ states that the labelling of $T$ 
can be turned (completed) into a correct encoding $\tilde{T}$ of a globally fulfilled decomposition tree 
$\cT$ that witnesses $\varphi$ (we call $\tilde{T}$ an \emph{expansion} of $T$). 

We observe that $\tilde{\varphi}_{\path}$ contains only {\sl positive} occurrences of modalities $\AG$, $\AF$, and $\AX$. 
Hence, by replacing all occurrences of $\AG$ (resp., $\AF$, $\AX$) in $\tilde{\varphi}_{\path}$  by $\mathbf{G}$  (resp., $\mathbf{F}$, $\mathbf{X}$) and by using standard techniques in automata theory, one can construct a {\sl deterministic} B\"uchi automaton over $\omega$-words $\cA_\path$ equivalent to $\tilde{\varphi}_\path$,
that is, such that $\tilde{\varphi}_\path$ holds over any expansion $\tilde{T}$ if and only if 
$\cA_\path$ accepts all paths $\pi$ of $\tilde{T}$.
$\cA_\path$ can be assumed to be deterministic because modalities $\mathbf{G}$ and $\mathbf{F}$ never occur under a negation and no occurrence of $\mathbf{G}$ is nested in an occurrence of $\mathbf{F}$
in the LTL formula 
corresponding to $\tilde{\varphi}_\path$. 
To avoid any exponential blowup in the construction of 
$\cA_\path$, 
one can use symbolic representations of states and transitions (or, equivalently, linear weak alternation \cite{alternating_automata_and_logics}). 
More precisely, states and transitions of $\cA_\path$
can respectively be represented by tuples of bits, each one corresponding to a subformula 
of $\tilde{\varphi}_\path$ that has to be evaluated,
and by propositional formulas over the bits of the source and target states and 
the input letters. 
Using techniques similar to those in \cite{LTL_to_symbolic_automata}, a symbolic representation of $\cA_\path$ can be computed directly from $\tilde{\varphi}_\path$ in polynomial time.

Now, if we project (the symbolic representation of) the deterministic B\"uchi automaton $\cA_\path$ onto the three proposition letters $0,1,2$, by discarding all other letters from the expansion $\tilde{T}$ of $T$, we obtain (a symbolic representation of) 
a non-deterministic B\"uchi automaton $\cA_\path^\exists$ that accepts all $\omega$-words from $\{0\}\cdot\{1,2\}^\omega$ 
if and only if $\cA_{\path}$ accepts all paths of {\sl some} expansion $\tilde{T}$ of $T$. Finally, the acceptance problem for $\cA_\path^\exists$ can be reduced to the universality problem for (symbolically represented) non-deterministic B\"uchi automata as follows. Let $\overline{\{0\}\cdot\{1,2\}^\omega} (= \{1,2\} \cdot \{0, 1,2\}^\omega)$ be 
the complement of the $\omega$-regular language $\{0\}\cdot\{1,2\}^\omega$. It holds that:
$$
\begin{array}{rl}
  \sL\big(\cA_{\path}^\exists\big) \cup (\{1,2\} \cdot \{0, 1,2\}^\omega) ~=~ \{0,1,2\}^\omega 
  & \; \text{iff} \;\; \sL\big(\cA_{\path}^\exists\big) ~\supseteq~ \{0\}\cdot\{1,2\}^\omega \\[1ex]
  & \; \text{iff} \;\; \ex{\tilde{\cT}} \; \tilde{\cT}\sat\tilde{\varphi}_{\tree} \,\et\, \fa{\pi} \; \tilde{\cT}|_\pi \in \sL\big(\cA_{\path}\big) \\[1ex]
  & \; \text{iff} \;\; \ex{\tilde{\cT}} \; \tilde{\cT}\sat\tilde{\varphi}_{\tree} \,\et\, \tilde{\cT}\sat\tilde{\varphi}_{\path} \\[1ex]
  & \; \text{iff} \;\; \ex{\tilde{\cT}} \; \tilde{\cT}\sat\tilde{\varphi}.
\end{array}
$$
It is not difficult to see that the universality problem for (symbolically represented) non-deterministic B\"uchi automata is in \pspace
(one can use a variant of Savitch's theorem \cite{complexity}). This provides a procedure to decide, in polynomial space, 
whether the Cone Logic formula $\varphi$ appears at the root of some globally fulfilled decomposition tree, and thus, thanks to 
Propositions \ref{prop:completeness} and \ref{prop:soundness}, whether $\varphi$ is satisfied by some labelled rational stripe.
\end{proof}
 \section{Cone Logic and modal logics of time intervals}\label{sec:applications}

In this section, we prove that Cone Logic subsumes an interesting and expressive temporal 
logic based on intervals and relations over them (a subset of the so-called Allen's relations). 
Interval temporal logics of Allen's relations  (the full logic HS and its fragments) have been 
originally introduced by Halpern and Shoham \cite{interval_modal_logic}. 
The basic elements of these logics are the intervals over a fixed, underlying temporal 
domain, e.g., $(\bbQ,<)$. Proposition letters are associated 
with intervals, and existential quantifications are guarded by some of the 12 possible non-trivial 
ordering relations between pairs of intervals \cite{interval_relations}, that is, the ``During'' or 
``sub-interval'' relation $D$, the ``Beginning'' relation $B$, the ``Ending'' relation $E$, the 
``Overlapping'' relation $O$, and so on. 

A number of results about the satisfiability problem for HS fragments have been given in the 
last years that mark the boundary between decidability and undecidability. The rule of thumb 
is that most interval temporal logics are undecidable. 
An up-to-date account of undecidability results for HS fragments can be found in \cite{itl_dark_side}.
Among the known results, we recall 
the undecidability of the logics $D$ (quantifying over sub-intervals)
and $O$ (quantifying over overlapping intervals) -- as well as of their transposes -- interpreted 
over infinite discrete temporal domains \cite{d_undecidable_journal,itl_dark_side}, and 
the undecidability of the logic $BE$ (quantifying over beginning and ending intervals) interpreted 
over both dense and infinite discrete temporal domains \cite{undecidability_BE,D_undec}.

Here we consider the fragment of HS that features the six modalities $\ang{D}$, $\ang{\bar{D}}$, 
$\ang{B}$, $\ang{\bar{B}}$, $\ang{L}$, and $\ang{\bar{L}}$, allowing one to quantify existentially over 
sub-intervals, super-intervals, beginning intervals, begun-by intervals, later intervals, and earlier 
intervals, respectively. 
We present a reduction from the satisfiability problem for $B\bar{B}D\bar{D}L\bar{L}$
to that for Cone Logic, thus proving that the former logic is decidable in polynomial
space when interpreted over the class of dense linear orders.
As a matter of fact, this result partially disproves a conjecture by Lodaya \cite{undecidability_BE}
concerning the undecidability of the satisfiability problem for the fragment $D\bar{D}$
-- strictly speaking, Lodaya did not specify whether the fragment $D\bar{D}$ was interpreted 
over discrete or dense temporal domains. In this respect, it is worth remarking that the
decidability of the HS fragments $B\bar{B}D\bar{D}L\bar{L}$, $D\bar{D}$, and $D$ depends 
on the class of temporal domains where these logics are interpreted.

\smallskip
As a preliminary step, we briefly introduce the syntax and the semantics of the logic $B\bar{B}D\bar{D}L\bar{L}$.
From now on, we assume the underlying temporal domain to be (isomorphic to) the linear ordering 
$(\bbQ,<)$ of the rational numbers and that intervals are non-singleton, closed convex subsets of 
such an ordering, namely, sets of the form $[x,y]=\settc{z\in\bbQ}{x\le z\le y}$, with $x,y\in\bbQ$ 
and $x<y$. We shortly denote by $\bbI$ the set of all intervals over $(\bbQ,<)$. Given $I=[x,y]$ 
and $I'=[x',y']$ in $\bbI$, if $x<x'<y'<y$, then we say that $I'$ is a (strict) \emph{sub-interval} 
of $I$ or, equivalently, that $I$ is a (strict) \emph{super-interval} of $I'$; similarly, if $x'=x$ and $y'<y$, 
then we say that $I'$ \emph{begins} $I$ or, equivalently, that $I$ is \emph{begun by} $I'$; finally, if $x'>y$, 
then we say that $I'$ is \emph{later} than $I$ or, equivalently, that $I$ is \emph{earlier} than $I'$.

Formulas of the logic $B\bar{B}D\bar{D}L\bar{L}$ are built up from proposition letters in a 
signature $\Sigma$ using the standard Boolean connectives and the modalities $\ang{D}$, $\ang{\bar{D}}$, 
$\ang{B}$, $\ang{\bar{B}}$, $\ang{L}$, and $\ang{\bar{L}}$, with the obvious semantics. For instance, given a 
labelled interval structure $\ang{\bbI,(R_a)_{a\in\Sigma}}$, where $R_a\subseteq\bbI$ for all $a\in\Sigma$, 
and given an initial interval $I$, we write $\ang{\bbI,(R_a)_{a\in\Sigma},I}\sat\ang{D}a$ if and only if there is a 
sub-interval $I'$ of $I$ such that $I'\in R_a$.

\smallskip
In the following, we prove that the logic $B\bar{B}D\bar{D}L\bar{L}$ has a decidable satisfiability problem
by translating its formulas into equi-satisfiable formulas of Cone Logic.
Such a translation exploits the existence a natural bijection between the intervals $I=[x,y]$ in $\bbI$ and 
the points $p=(x,y)$ in the rational plane such that $x<y$ (hereafter, we call these points \emph{interval-points}). 

The first step it to show that the region of all interval-points can be somehow described by a formula 
of Cone Logic. Let $\positive$, $\negative$, $\singleton$ be three fresh proposition letters and let $\psi_0$ 
be the following formula of Cone Logic in the signature $\Sigma'=\Sigma\cup\{\positive,\negative,\singleton\}$:
$$
\begin{array}{rcl}
  \psi_0 &~~=&~~  \B(\positive \vel \negative \vel \singleton) \\[1ex]
         &~~\et&~~ \B(\neg\positive \vel \neg\negative) ~\et~ 
                  \B(\neg\positive \vel \neg\singleton) ~\et~ 
                  \B(\neg\negative \vel \neg\singleton)  \\[1ex]
         &~~\et&~~ \B\DTstrict\DBstrict\,\singleton ~\et~ 
                  \B(\singleton \then \BULp\positive \et \BTstrict\positive \et 
                                      \BLRp\negative \et \BBstrict\negative).
\end{array}
$$
Consider now a labelled rational plane $\cP=\ang{\bbP,(R_a)_{a\in\Sigma'}}$ that satisfies $\psi_0$
(see Figure \ref{fig:intervalpoints}). Clearly, the three regions $R_\positive$, $R_\negative$,
and $R_\singleton$ form a partition of the entire plane $\bbP$ (this is enforced by the first two 
lines of $\psi_0$). Moreover, the region $R_\singleton$ has the form of a trajectory $y=f(x)$ that is 
``almost a diagonal'', in the sense that for every $x\in\bbQ$ there is exactly one $y\in\bbQ$ such 
that $(x,y)\in R_\singleton$ and all other points of $R_\singleton$ are contained in the lower-left 
quadrant and in the upper-right quadrant centred at $(x,y)$. In general, the region $R_\singleton$ 
might not coincide with the diagonal $\settc{(x,x)}{x\in\bbQ}$ -- note that if this happens, we would 
immediately have that $R_\positive$ contains all and only the interval-points, that is, the points
$(x,y)\in\bbP$, with $x<y$. Nonetheless, we can prove the following lemma.
\begin{figure}[!!t]
\centering
\includegraphics[scale=0.9]{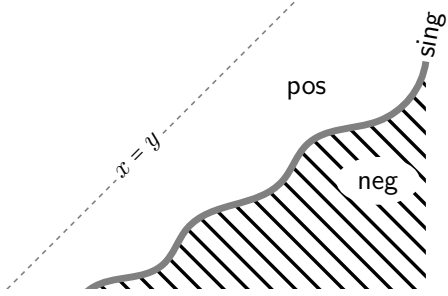}
\caption{A $\positive$-labelled region delimiting (pseudo-)interval-points.}
\label{fig:intervalpoints}
\end{figure}

\begin{lem}\label{lemma:intervalregion}
Given a formula $\varphi$ of Cone Logic, if $\cP=\ang{\bbP,(R_a)_{a\in\Sigma'}}$ is a labelled rational plane
that satisfies $\varphi\et\psi_0$, then there is a labelled rational plane $\cP'=\ang{\bbP,(R'_a)_{a\in\Sigma'}}$ 
that still satisfies $\varphi\et\psi_0$ and such that $R'_\singleton=\settc{(x,x)}{x\in\bbQ}$.
\end{lem}

As a preliminary step, we prove that we can ``stretch'' the $y$-coordinates of a labelled plane $\cP$ with respect to any 
strictly increasing function $f$, without affecting the satisfaction of any of the subformulas:

\begin{clm}
For every strictly increasing function $f:\bbQ\then\bbQ$, the labelled plane $\cP=\ang{\bbP,(R_a)_{a\in\Sigma'}}$ 
has the same shading as the labelled rational plane $f(\cP)=\bigang{f(\bbP),\big(f(R_a)\big)_{a\in\Sigma'}}$, where 
$f(R)=\bigsettc{(x,f(y))}{(x,y)\in R}$ for all $R\subseteq\bbP$.
\end{clm}

\begin{proof}[Proof of the claim]
To start with, we observe that $f(\bbQ)$, equipped with the natural ordering of the rational numbers, 
is a countable dense linear order with no minimum nor maximum elements. Hence $f(\cP)$ can be given the 
status of labelled rational plane. 

To conclude the proof, it suffices to observe that for all points $(x,y),(x',y')\in\bbP$ and all spatial relations $\Carg{d}$, 
$$
  (x,y) ~\Carg{d}~ (x',y') \qquad\text{iff}\qquad (x,f(y)) ~\Carg{d}~ (x',f(y')).
$$
Using to the view-to-type dependency, we derive $\type_{\cP}(x,y)=\type_{f(\cP)}(x,f(y))$
for all points $(x,y)\in\bbP$, which shows that $\type_\cP(\bbP)=\type_{f(\cP)}(f(\bbP))$.
\end{proof}

\begin{proof}[Proof of Lemma \ref{lemma:intervalregion}] 
Let $\cP=\ang{\bbP,(R_a)_{a\in\Sigma'}}$ be a model for the formula $\varphi\et\psi_0$. 
By the definition of $\psi_0$, 
there is a function $f:\bbQ\then\bbQ$ such that for all $x\in\bbQ$, $(x,y)\in R_\singleton$ if and only
if $y=f(x)$ (this is enforced by the third line of the definition of $\psi_0$). 
By the definition of $\psi_0$, it holds that 
$f$ is strictly increasing: if $x<x'$, then $(x,f(x))$ and $(x',f(x'))$ are two points in $R_\singleton$ such that
$(x,f(x)) ~\CUR~ (x',f(x'))$, and thus $f(x)<f(x')$.

Now, let us denote by $f^{-1}$ the inverse of the function $f$, which is also strictly increasing. By the previous 
claim, 
we know that the ``stretched'' labelled plane $\cP'=f^{-1}(\cP)$, which is obtained by mapping 
each point $(x,y)$ of $\cP$ to the point $(x,f^{-1}(y))$, has the same shading as $\cP$, and hence it 
also satisfies the formula $\varphi\et\psi_0$. 
Moreover, by construction, the region of all $\singleton$-labelled points in $\cP'$ coincides with the 
diagonal of $\cP'$:
$$
  f^{-1}(R_\singleton) ~=~ \bigsettc{(x,f^{-1}(y))}{(x,y)\in R_\singleton} ~=~ \bigsettc{(x,f^{-1}(f(x))}{x\in\bbQ}.
$$
This shows that the $\positive$-labelled points of $\cP'$ are exactly the interval-points.
\end{proof}

\medskip
Making use of Lemma \ref{lemma:intervalregion}, we can translate any 
formula $\varphi$ of the logic $B\bar{B}D\bar{D}L\bar{L}$ into an equi-satisfiable formula 
$\tilde{\varphi}$ of Cone Logic,
which is obtained by first replacing each occurrence of a subformula $\ang{D}\alpha$ (resp., 
$\ang{\bar{D}}\alpha$, $\ang{B}\alpha$, $\ang{\bar{B}}\alpha$, $\ang{L}\alpha$, $\ang{\bar{L}}
\alpha$) in $\varphi$  by the formula $\DLR(\positive\et\alpha)$ (resp., $\DUL(\positive\et\alpha)$, $\DBstrict(\positive\et\alpha)$, $\DTstrict(\positive\et\alpha)$, $\BLRp(\positive\then\DUR(\positive\et\alpha))$, $\BLRp(\positive\then\DLL(\positive\et\alpha))$) and then adding the conjunct $\psi_0$. 

We can easily check the correctness of the translation for modalities
$\ang{D}$, $\ang{\bar{D}}$, $\ang{B}$, and $\ang{\bar{B}}$. 
Proving that the translation of modalities $\ang{L}$ and $\ang{\bar{L}}$ is correct as well
is less straightforward.
Let us consider an interval $I=[x,y]$ and a later interval $I'=[x',y']$ of $I$, with $x'>y$. 
Figure \ref{fig:future} depicts the spatial relationships between the corresponding interval-points 
$p=(x,y)$ and $p'=(x',y')$ and the intermediate point $q=(y,y)$. Clearly, for every interval-point 
$q'$ such that $p \,\CLRp\, q'$, we have $q' \,\CURp\, q$ 
and $q \,\CUR\, p'$, 
and hence $q' \,\CUR\, p'$. Conversely, if $p=(x,y)$ and $p'=(x',y')$ are two 
interval-points such that $p \,\CLRp\, q'$ implies $q' \,\CUR\, p'$ for all interval-points $q'$, 
then we necessarily have $x'>y$ and hence $I'=[x',y']$ is a later interval of $I=[x,y]$. This shows 
that the translation that replaces each occurrence of a subformula $\ang{L}\alpha$ by the formula
$\BLRp(\positive\then\DUR(\positive\et\alpha))$ is correct. Similar arguments can be used to prove the 
correctness of the translation for $\ang{\bar{L}}$.
\begin{figure}[!!t]
\centering
\includegraphics[scale=0.9]{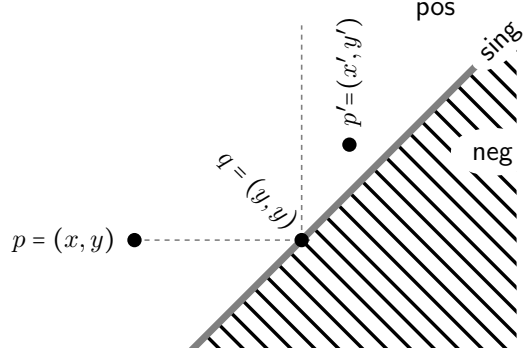}
\caption{Spatial relationship between an interval-point and its future.}
\label{fig:future}
\end{figure}

Now, the translation of $B\bar{B}D\bar{D}L\bar{L}$ formulas into equi-satisfiable Cone Logic formulas, 
together with the decidability result given in Section \ref{sec:solution}, yields a \pspace procedure 
to decide the satisfiability problem for the logic $B\bar{B}D\bar{D}L\bar{L}$ when the underlying 
domain is assumed to be dense. This subsumes previous results from \cite{subinterval_tableau_journal}. 
Moreover, we know from \cite{chronological_future_modality} that the satisfiability problem for the 
interval logic that only uses the subinterval operator $\ang{D}$ is already \pspace-hard over
dense domains. We can thus sum up with the following theorem.

\begin{thm}\label{theo:pspace}
The satisfiability problem of Cone Logic over the rational plane and that of the 
interval temporal logic $B\bar{B}D\bar{D}L\bar{L}$ over $\bbQ$ are \pspace-complete.
\end{thm}

Thanks to the above theorem and Remark \ref{rem:rationalplanes}, we know that the satisfiability 
problem for Cone Logic interpreted over the class of all labelled planes (including the rational and the real labelled planes) and that of the interval temporal logic $B\bar{B}D\bar{D}L\bar{L}$ interpreted over the class of all 
dense temporal domains are \pspace-complete.
Finally, we point out that similar decidability results hold for the logic $E\bar{E}D\bar{D}L\bar{L}$, 
by simply changing the orientation of the $x$- and $y$-axes.

 \section{Conclusions}\label{sec:conclusions}

In this paper, we investigated the satisfiability problem for a suitable weakening of 
Venema's Compass Logic, called Cone Logic, and we proved that, unlike the cases of
Compass Logic and other projection-based spatial logics, it is decidable (\pspace-complete)
over the rational plane $\bbQ\times\bbQ$. 
Moreover, we showed that such a decidability result can be exploited to prove the
decidability of the interval temporal logic $B\bar{B}D\bar{D}L\bar{L}$ of Allen's 
relations `Begins', `During', and `Later', and their transposes, over the class of 
dense linear orders (equivalently, the rational numbers), thus disproving a
conjecture by Lodaya \cite{undecidability_BE}.

One may consider possible extensions of Cone Logic in various directions. 
For instance, one may consider multi-dimensional spaces and introduce a corresponding 
logic to describe spatial relationships over points in these spaces (in general, $2^n$ 
distinct cone-shaped directions exist in a space with $n$ dimensions). Alternatively, 
one may partition the two-dimensional space into more than four cone-shaped directions. 
In all such cases, we believe it possible to generalize the achieved results, e.g., the tree 
pseudo-model property, in a rather natural way (the complexity of the satisfiability 
problem, however, may increase significantly). 
Other generalizations of Cone Logic envisage the use of region-based relationships. 
As an example, the correspondence between intervals over the rational line and points 
over the rational plane can be lifted to higher-dimensional objects, proving, for instance, 
that a suitable spatial logic based on rectangular regions, that is, 
two-dimensional intervals, 
is subsumed by a 
four-dimensional point-based modal logic very similar to Cone Logic. This 
would establish a first interesting bridge between Cone Logic and relativistic temporal logics 
based on Minkowski's space-time structure \cite{chronological_future_modality}.

The most interesting open problem is that of determining whether or not
Cone Logic remains decidable when interpreted over the real plane $\bbR\times\bbR$.
In Remark \ref{rem:rationalplanes}, we have seen that, if a Cone Logic formula holds over 
$\bbR\times\bbR$, then it also holds over $\bbQ\times\bbQ$. The converse does not hold 
in general, as there exist formulas of Cone Logic, e.g., that of Example \ref{ex:cantor}, that 
hold over $\bbQ\times\bbQ$, but not over $\bbR\times\bbR$. 
The satisfiability problem for Cone Logic over $\bbR\times\bbR$ is not known to be decidable, 
and the same applies to the interval temporal logic $B\bar{B}D\bar{D}L\bar{L}$ interpreted 
over $\bbR$. We plan to study these decidability problems in the future.

\smallskip\noindent
{\bf Acknowledgments.}\
The authors would like to thank the anonymous reviewers and Dario Della Monica for their 
valuable comments on a preliminary draft of this paper. 
\bibliographystyle{plain}
\bibliography{main}

\end{document}